\documentclass{CMSLATEX}
\usepackage{latexsym, amssymb, enumerate, amsmath}

\usepackage{graphicx}
\renewcommand{\theequation}{\arabic{section}.\arabic{equation}}

\makeatletter
    
    \newcommand{\Rmnum}[1]{\expandafter\@slowromancap\romannumeral #1@}
  \makeatother

\def\({\left(}
\def\){\right)}
\def\[{\left[}
\def\]{\right]}

\newtheorem{thm}{Theorem}[section]
\newtheorem{prop}[thm]{Proposition}
\newtheorem{lem}[thm]{Lemma}
\newtheorem{cor}[thm]{Corollary}

\newtheorem{defn}[thm]{Definition}

\newtheorem{rem}[thm]{Remark}

\begin{document}

\title{Conservation Laws of Random Matrix Theory}
\author{Nicholas M. Ercolani
\thanks {Department of Mathematics, The University of Arizona, Tucson, AZ 85721--0089, ({\tt ercolani@math.arizona.edu}). Supported by NSF grant DMS-0808059.}}


\pagestyle{myheadings} \markboth{CONSERVATION LAWS OF RANDOM MATRIX THEORY}{
N. M. ERCOLANI}\maketitle

\begin{abstract}
This paper presents an overview of the derivation and significance of recently derived conservation laws for the matrix moments of Hermitean random matrices with dominant exponential weights that may be either even or odd. This is based on a detailed asymptotic analysis of the partition function for these unitary ensembles and their scaling limits. As a particular application we derive closed form expressions 
for the coefficients of the genus expansion for the associated free energy  in a particular class of dominant even weights. These coefficients are generating functions for enumerating {\it g}-maps, related to graphical combinatorics on Riemann surfaces. This generalizes and resolves a 30+ year old conjecture in the physics literature related to quantum gravity.
\end{abstract}




\section{Introduction} \label{sec:1}

This paper will present an overview of some recent developments in the application of random matrix analysis to the {\it topological combinatorics} of surfaces. Such applications have a long history about which we should say a few words at the outset. The combinatorial objects of interest here are {\it maps}. A map is an embedding of a graph into a compact, oriented and connected  surface $X$ with the requirement that the complement of the graph in $X$ should be a disjoint union of simply connected open sets. If the genus of $X$ is $g$, this object is referred to as a {\it g-map}.
The notion of $g$-maps was introduced by Tutte and his collaborators in the '60s \cite{Tu} as part of their investigations of the four color conjecture.

In the early '80s Bessis, Itzykson and Zuber, a group of physicists studying 't Hooft's diagrammatic  approaches to large N expansions in quantum field theory,  discovered a profound connection between the problem of enumerating $g$-maps and random matrix theory \cite{BIZ80}. That seminal work was the basis for bringing asymptotic analytical methods into the study of maps and other related combinatorial problems. 

Subsequently, in the early '90s, other physicists \cite{DoS, GrMi} realized that the 
matrix model diagrammatics described in \cite{BIZ80} provide a natural means for discretizing the Einstein-Hilbert action in two dimensions. From that and a formal double scaling limit, they were able to put forward a candidate for so-called {\it 2D Quantum Gravity}. This generated a great deal of interest in the emerging field of string theory. We refer to  \cite{dFGZ} for a systematic review of this activity and to \cite{Mar} for a description of more recent developments related to topological string theory.  

All of these applications were based on the postulated existence of a $1/n^2$ asymptotic expansion of the free energy associated to the random matrix partition function, where $n$ denotes the size of the matrix, as $n$ becomes large. The combinatorial significance of this expansion is that the coefficient of $1/n^{2g}$ should be the generating function for the enumeration of $g$-maps (ordered by the cardinality of the map's vertices). In \cite{EM03} the existence of this asymptotic expansion and several of its important analytical properties were rigorously established. This analysis was based on a Riemann-Hilbert problem originally introduced by Fokas, Its and Kitaev \cite{FIKII} to study the 2D gravity problem. 

The aim of this paper is to outline how the results of \cite{EM03} and its sequel \cite{EMP08} have been used to gain new insights into the map enumeration problem. In particular, we will be able to prove and significantly extend a conjecture made in \cite{BIZ80} about the closed form structure of the generating functions for map enumeration. 

Over time combinatorialists have made novel use of many tools from analysis including contour integrals and differential equations. In this work we also introduce nonlinear partial diferential equations, in particular a hierarchy of conservation laws reminiscent of the {\it shallow water wave equations} \cite{Wh} (see (\ref{TODA})). This appears to make contact with the class of {\it differential posets} introduced by Stanley \cite{St} (see Remark \ref{diffposet}). 

\section{Background}
The general class of matrix ensembles we analyze has probability measures of the form
\begin{eqnarray} \label{RMT}
d\mu_{t_{j}} &=& \frac{1}{{Z}^{(n)}(g_s, t_{j})}\exp\left\{-\frac1{g_s} \mbox{ Tr} [V_j(M, t_{j})]\right\} dM,\,\, \mbox{where}\\
\label{I.001b} V_j(\lambda; \ t_{j} ) &=&  \frac{1}{2} \lambda^{2} +  \frac{t_{j}}{j} \lambda^{j}
\end{eqnarray}
defined on the space $\mathcal{H}_n$ of $n \times n$ Hermitean matrices, $M$, and with $g_s$ a positive parameter, referred to as the {\it string coefficient}. The normalization factor ${Z}^{(n)}(g_s, t_{j})$, which serves to make $\mu_t$ a probability measure, is called the {\it partition function} of this unitary ensemble.
\begin{rem}
In previous treatments, \cite{EM03, EMP08, Er09, EP11}, we have used the parameter $1/N$ instead of $g_s$. This was in keeping with notational usages in some areas of random matrix theory; however, since here we are trying to make a connection to some applications in quatum gravity, we have adopted the notation traditionally used in that context. This also is why we have scaled the time parameter $t_j$ by $1/j$ in this paper.
\end{rem}

For general polynomial weights $V$ it is possible to establish the following fundamental asymptotic expansion \cite{EM03}, \cite{EMP08} of the logarithm of the {\em free energy} associated to the partition function.  More precisely, those papers consider weights of the form 
\begin{eqnarray}\label{genpot}
V(\lambda) &=& \frac{1}{2} \lambda^{2} + \sum_{\ell=1}^{j} \frac{t_{\ell}}{\ell} \lambda^\ell.
\end{eqnarray} with $j$ even.

\noindent We introduce a renormalized partition function, which we refer to as a {\it tau function} representation,
\begin{equation} \label{tausquare}
\tau^2_{n,g_s}(\vec{t}\;) =  \frac{Z^{(n)}(\vec{t}\;)}{Z^{(n)}(g_s, 0)},
\end{equation}
where $\vec{t} = (t_1, \dots t_{j}) \in \mathbb{R}^j$. The principal object of interest is the {\em large n} asymptotic expansion of this representation for which one has the result \cite{EM03, EMP08}

\begin{eqnarray}
\label{I.002} \ \ \ \log \tau^2_{n,g_s}(\vec{t}\;) =
n^{2} e_{0}(x, \vec{t}\;) + e_{1}(x, \vec{t}\;) + \frac{1}{n^{2}} e_{2}(x, \vec{t}\;) + \cdots +  \frac{1}{n^{2g-2}} e_{g}(x, \vec{t}\;) + \dots
\end{eqnarray}
as $n \to \infty$ while $g_s \to 0$ with $x = n g_s$, called the 't Hooft parameter, held fixed. Moreover, for $\mathcal{T} =  (1-\epsilon, 1 + \epsilon) 
\times \left(\{|\vec{t}\;| < \delta\} \cap \{t_{j}> 0\}\right) $ for some $\epsilon > 0, \delta > 0$,
\begin{enumerate}[(i)]

\item \label{unif} the expansion is uniformly valid on compact subsets of $\mathcal{T}$;

\item \label{analyt} $e_g(x, \vec{t}\;)$ extends to be complex analytic in 
$\mathcal{T}^{\mathbb{C}} =\left\{(x,\vec{t}\;) \in \mathbb{C}^{j+1} \big| |x - 1| < \epsilon, |\vec{t}| < \delta\right\}$;

\item \label{diff} the expansion may be differentiated term by term in $(x,\vec{t}\;)$ with uniform error estimates as in (\ref{unif})

\end{enumerate}
The meaning of (\ref{unif}) is that for each $g$ there is a constant, $K_g$, depending only on $\mathcal{T}$  and $g$ such that 
\begin{equation*}
\left| \log \tau^2_{n,g_s} \left(\vec{t}\;\right) - n^2 e_0(x, \vec{t}\;)  -  \dots  - \frac{1}{n^{2g-2}} e_{g}(x, \vec{t}\;) \right|  \leq \frac{K_g}{n^{2g}}
\end{equation*}
for $(x, \vec{t}\;)$ in a compact subset of $\mathcal{T}$. The estimates referred to in (\ref{diff}) have a similar form with $\tau^2_{n,g_s}$ and $e_{j}(x, \vec{t}\;)$ replaced by their mixed derivatives (the same derivatives in each term) and with a possibly different set of constants $K_g$.
\begin{rem}
Recently these results were extended to the case where $j$ is odd \cite{EP11}.  In this case one should replace the normalized partiiton function (\ref{tausquare}) by its Szeg\"o representation in terms of eignevalues (\ref{szego}).
\end{rem}
\smallskip

To explain the topological significance of the $e_g(x,\vec{t})$ as generating functions, we begin with a precise definition of the objects they enumerate.
A {\it map} $\Sigma$ on a compact, oriented and connected surface $X$ is a
pair $\Sigma = (K(\Sigma), [\imath])$ where
\begin{itemize}
\item $K(\Sigma)$ is a connected 1-complex; 
\item $[\imath]$ is an
isotopical class of inclusions $\imath:K(\Sigma) \rightarrow X$; 
\item the complement of $K(\Sigma)$ in $X$ is a disjoint union of open cells
(faces); 
\item the complement of the vertices in $K(\Sigma)$ is
a disjoint union of open segments (edges).
\end{itemize}  
When the genus of X is $g$ one refers to the map as a $g-map$. What \cite{BIZ80} effectively showed was that the partial derivatives of $e_g(1,\vec{t})$ evaluated at $\vec{t} = 0$ "count" a geometric quotient of a certain class of {\it labelled g-maps}.  

As a means to reduce from enumerating these labelled $g$-maps to enumerating $g$-maps, it is natural to try taking a geometric quotient by a "relabelling group" more properly referred to as a {\it cartographic group} \cite{BI96}. 

This labelling has two parts; first the vertices of the same valence, $\ell$ have an order labelling $1, \dots n_\ell$ and second at each vertex one of the edges is distinguished. Given that X is oriented, this second labelling gives a unique ordering of the edges around each vertex. The fact that the coefficients of the free energy expansion (\ref{I.002}) enumerate this class of labelled $g$-maps is a consequence of (\ref{I.002}) (i)  which enables one to evaluate a mixed partial derivative of $e_g$ in terms of the {\it Gaussian Unitary Ensemble} (GUE) where correlation functions of matrix coefficients all reduce to two point functions. (A precise description of this correspondence may be found in \cite{EM03}). 

To help fix these ideas we consider the case of a $j$-regular $g$-map (i.e., every vertex has the same valence, $j$) of size $m$ (i.e., the map has $m$ vertices) which is the main interest of this paper. The cartographic group in this case is generated by the symmmetric group $S_m$ which permutes the vertex labels and $m$ factors of the cyclic group $C_{j},$ which rotates the distinguished edge at a given vertex in the direction of the holomorphic (counter-clockwise) orientation on $X$. The order of the cartographic group here is the same as that of the product of its factors which is $m! j^m$. On the other hand the generating function for $g$-maps in this setting is given by 
\begin{eqnarray} \label{gquotI}
e_g(t_j) &=& e_g\left(x = 1 , \vec{t} = (0, \dots, 0, t_j) \right)\\
&=& \sum_{m \geq 1}\frac{1}{m!j^m} (-t_j)^{m}\kappa^{(g)}_j(m)
\end{eqnarray}
where $\kappa^{(g)}_j(m)$ = the number of labelled $j$-regular $g$-maps on $m$ vertices. The factor $\frac{1}{m!j^m}$ perfectly cancels the order of the cartographic group, making this series appear to indeed be the ordinary generating function for pure $g$-maps.
However, for some $g$-maps the cartographic action may have non-trivial isotropy and this can create an "over-cancellation" of the labelling. This happens when a particular relabelling of a given map can be transformed back to the original labelling by a diffeomorphism of the underlying Riemann surface $X$. In this event the two labellings are indistinguishable and the diffeomorphism induces an automorphism of the underlying map. In addition, the element of the cartographic group giving rise to this situation is an element of the isotropy group of the given map. Hence, as a generating function for the geometric quotient, (\ref{gquotI}) is expressible as 
\begin{eqnarray} \label{gquotII}
e_g(t_j) &=& \sum_{g-\mbox{maps}\,\, \Sigma} \frac{1}{|\mbox{Aut}(\Sigma)|} (-t_j)^{m(\Sigma)}\\
\label{gquotIII} 
E_g(x, t_j) &=& e_g\left(x, \vec{t} = (0, \dots, 0, t_j) \right)\\
\nonumber &=& \sum_{g-\mbox{maps}\,\, \Sigma} \frac{1}{|\mbox{Aut}(\Sigma)|} (-t_j)^{m(\Sigma)} x^{f(\Sigma)}\\
\nonumber &=& x^{2-2g} e_g(x^{j/2 -1}t_j)
\end{eqnarray}
where $m(\Sigma)$ = the number of vertices of $\Sigma$, $f(\Sigma)$ = the number of faces of $\Sigma$ and Aut($\Sigma$) = the automorphism group of the map $\Sigma$. We have included the $x$-dependent form, (\ref{gquotIII}), of $e_g$ since that will play an important role later on and also to observe that this is in fact a {\it bivariate} generating function for enumerating $g$-maps with a fised number of vertices and faces. Moreover, in this $j$-regular setting, one sees that the bivariate function is self-similar.  This is a direct consequence of Euler's relation:
\begin{eqnarray} \nonumber
2 -2g &=& \# \,\mbox{vertices}\, - \# \,\mbox{edges}\, +  \#\, \mbox{faces} \\
\nonumber &=& m(\Sigma) - \frac j2 m(\Sigma) + f(\Sigma). \\
\label{sss} t_j^{m(\Sigma)} x^{f(\Sigma)} &=& x^{2-2g}(t_j x^{j/2 -1})^{m(\Sigma)}.
\end{eqnarray}
The presence of geometric factors such as $\frac{1}{|\mbox{Aut}(\Sigma)|}$ is not uncommon in enumerative graph theory, a classical example being that of Erd\"os-R\'enyi graphs \cite{JKLP}. In the quantum gravity setting these factors also have a natural interpretation in terms of the discretization of the reduction to conformal structures via a quotient of metrics by the action of the diffeomorphism group. We refer to \cite{dFGZ, BI96} for further details on this attractive set of ideas.

In \cite{BIZ80}, $e_0, e_1$ and $e_2$ were explicitly computed for the case of valence $j=4$. We quote, from the same paper, the following conjecture (some notation has been changed to be consistent with ours):
\smallskip

{\it``It would of course be very interesting to obtain $e_g(t_4)$ in closed form for any value of $g$. The method of this paper enabled us to do so up to $g=2$, but works in the general case, although it requires an incresing amount of work. We conjecture a general expression of the form}
\begin{eqnarray*}
e_g = \frac{(1-z_0)^{2g-1}}{(2-z_0)^{5(g-1)}}P^{(g)}(z_0), \qquad g\geq 2
\end{eqnarray*}
{\it with $P^{(g)}$ a polynomial in $z_0$, the degree of which could be obtained by a careful analysis of the above procedure."}

\noindent Here $z_0 = z_0(t_4)$ is equal, up to a scaling, to the generating function for the Catalan numbers; below it will signify $z_0(t_{2\nu})$ which is similarly related to the generating function for the higher Catalan numbers (\ref{catalan}). 

Over the years there have been a number of attempts to systematically address this question by studying the resolvent of the random matrix and associated Schwinger-Dyson equations \cite{A, Eynard}. Our methods take a different approach.
\smallskip

The main purpose of this paper is to show how this conjecture can be verified and significantly extended. In particular, we will show that for the case of even valence, $j = 2\nu$,

\begin{thm} \label{thm51} For $g \geq 2$,
\begin{eqnarray}
\label{note} e_g(z_0) &=& C^{(g)} + \frac{c_0^{(g)}(\nu)}{(\nu - (\nu-1)z_0)^{2g-2}} + \cdots + \frac{c_{3g-3}^{(g)}(\nu)}{(\nu - (\nu-1)z_0)^{5g-5}}\\
\label{note2}&=&  \frac{(z_0 - 1)^r Q_{5g-5-r}(z_0)}{(\nu - (\nu - 1)z_0)^{5g-5}}\, ,\\
\nonumber r &=& \max\left\{ 1, \left\lfloor\frac{2g-1}{\nu-1}\right\rfloor\right\}\, ,
\end{eqnarray}
for all $\nu \geq 2$.  The top coefficient and the constant term are respectively given by
\begin{eqnarray} \label{leadcoeff} 
c_{3g-3}^{(g)}(\nu) &=& \frac{g!}{(5g-5)(5g-3) \nu^2} a_{3g-1}^{(g)}(\nu)\ne 0\, ,\\
 \label{note3} C^{(g)} = &  -2  (2g-3)!& \left[\frac{1}{(2g+2)!} - \frac{1}{(2g)! 12} + \frac{(1 - \delta_{2,g})}{(2g-1)!}\sum_{k=2}^{g-1} \frac{(2-2k)_{2g-2k+2}}{(2g-2k+2)!} C^{(k)}\right] \qquad
\end{eqnarray}
where $a_{3g-1}^{(g)}(\nu)$ (see Thm \ref{result}) is proportional to the $g^{th}$ coefficient in the asymptotic expansion at infinity of the $\nu^{th}$ equation in the Painlev\'e I hierarchy \cite{Er09} and $(r)_m = r(r-1)\dots(r-m+1)$.
\end{thm}

\noindent Our methods can be extended to the case of $j$ odd and the derivation of the analogue to Theorem \ref{thm51} is in progress (see section \ref{sec:6}). 

The route to getting these results passes through nonlinear PDE, in particular a class of nonlinear evolution equations known as conservation laws which come from studying scaling limits of the recursion operators for orthogonal polynomials whose weights match those of the matrix models. 

This appeal to orthogonal polynomials also motivated the approaches of \cite{BIZ80} and \cite{DoS}. However to give a rigorous {\it and} effective treatment to the problem of finding closed form expressions for the coefficients of the asymptotic free energy, (\ref{I.002}), requires essential use of
Riemann-Hilbert analysis on the Riemann Hilbert Problem for orthogonal polynomials that was introduced in \cite{FIKII}. Though we will not review this analysis here, we will state the consequences of it needed for our applications and reference their sources.  

In section \ref{sec:2} we present the necessary background on orthogonal polynomials and introduce the main equations governing their recurrences operators: the {\it difference string equations} and the {\it Toda lattice equations}. In  section \ref{sec:44} we describe how  (\ref{I.002}) can be used to derive and solve (in the case of even valence) the continuum limits of these equations which relates to the nonlinear evolution equations alluded to earlier. In section \ref{sec:4} we outline the proof Theorem \ref{thm51} and in section \ref{sec:6} we describe the extension of this program to the case of odd valence and briefly mention what has been accomplished in that case thus far. This will also help to illuminate the full picture behind the idea of conservation laws for random matrices.

\section{The Role of Orthogonal Polynomials and their Asymptotics} \label{sec:2}

Let us recall the classical relation between orthogonal polynomials and the space of square-integrable functions on the real line, $\mathbb{R}$, with respect to exponentially weighted measures. In particular, we want to focus attention on weights that correspond to the random matrix weights, $V(\lambda)$, (\ref{I.001b}) with $j$ even. (Recently this relation has been extended to the cases of $j$ odd \cite{BD10, EP11}, with the orthogonal polynomials generalised to the class of so-called {\em non-Hermitean} orthogonal polynomials; however, for this exposition we will stick primarily with the even case.)  To that end we consider the Hilbert space $H = L^2\left(\mathbb{R}, e^{-g_s^{-1}V(\lambda)}\right)$ of weighted square integrable functions. This space has a natural polynomial basis,
$\{\pi_n(\lambda)\}$, determined by the conditions that
\begin{eqnarray*}
\pi_n(\lambda) &=& \lambda^n + \,\,\mbox{lower order terms}\\
\int \pi_n(\lambda) \pi_m(\lambda) e^{-g_s^{-1}V(\lambda)} d\lambda &=& 0\,\, \mbox{for}\,\, n \ne m.
\end{eqnarray*}
For the construction of this basis and related details we refer the reader to \cite{Deift}.

With respect to this basis, the operator of multiplication by $\lambda$ is representable as a semi-infinite tri-diagonal matrix, 
\begin{equation} \label{multop}
\mathcal{L} = \begin{pmatrix} a_0 & 1 &  \\
                              b^2_1 & a_1 & 1 \\
			          & b^2_2 & a_2  & \ddots \\
                                  &      & \ddots & \ddots 
\end{pmatrix}\,.
\end{equation}
$\mathcal{L}$ is commonly referred as the {\it recursion operator} for the orthogonal polynomials and its entries as {\it recursion coefficients}. (When $V$ is an even potential, it follows from symmetry that $a_j = 0$ for all $j$.) We remark that often a basis of orthonormal, rather than monic orthogonal, polynomials is used to make this representation.  In that case the 
analogue of (\ref{multop}) is a symmetric tri-diagonal matrix. As long as the coefficients $\{b_n\}$ do not vanish, these two matrix representations can be related through conjugation by a semi-infinite diagonal matrix of the form $\mbox{diag}\,(1, b^{-1}_1, \left(b_1 b_2\right)^{-1}, \left(b_1 b_2 b_3 \right)^{-1}, \dots)$. 

Similarly, the operator of differentiation with respect to $\lambda$, which is densely defined on ${H}$, has a semi-infinite matrix  representation, 
$\mathcal{D}$, that can be expressed in terms of $\mathcal{L}$ as
\begin{eqnarray}
\label{diffrep2}\mathcal{D} &=&  \frac1{g_s} \left(\mathcal{L} + t \mathcal{L}^{j-1}\right)_-
\end{eqnarray}
where the "minus" subscript denotes projection onto the strictly lower part of the matrix. 

From the canonical (Heisenberg) relation on $H$, one sees that
\begin{eqnarray*}
\left[\partial_\lambda, \lambda\right] &=& 1,
\end{eqnarray*}
where here $\lambda$ in the bracket  and $1$ on the right hand side are regarded as multiplication operators. With respect to the basis of orthogonal polynomials this may be re-expressed as
\begin{eqnarray} \label{fl1} 
\left[ \mathcal{L}, \left(\mathcal{L} +  t \mathcal{L}^{j-1}\right)_-  \right] &=& g_s I .
\end{eqnarray}
The relations implicit in (\ref{fl1}) have been referred to as {\it string equations} in the physics literature.  In fact the relations that one has, row by row, in (\ref{fl1}) are actually successive differences of consecutive string equations in the usual sense. However, by continuing back to the first row one may recursively de-couple these differences to get the usual equations. To make this distinction clear we will refer to the row by row equations that one has directly from (\ref{fl1}) as {\it difference string equations}.

$\mathcal{L}$ depends smoothly on the coupling parameter $t_j$ in the weight $V(\lambda)$ (see \ref{I.001b}). The explicit dependence can be determined from the fact that multiplication by $\lambda$ commutes with differentiation by $t_j$. This yields our second fundamental relation on the recurrence coefficients,
\begin{eqnarray} \label{fl2}
g_s \frac{\partial}{\partial t_{j}}\mathcal{L} &=& \left[\left(\mathcal{L}^{j}\right)_- , \mathcal{L}\right]\,,
\end{eqnarray}
which is equivalent to  the $j^{th}$ equation of the semi-infinite Toda Lattice hierarchy.
The Toda equations for $j = 1$ are
\begin{eqnarray}\label{T1a}
- g_s \frac{da_{n,g_s}}{dt_1} &=& b^2_{n+1,g_s} - b^2_{n,g_s}\\
\label{T1b} - g_s \frac{db^2_{n,g_s}}{dt_1} &=& b^2_{n,g_s} \left(a_{n,g_s} - a_{n-1, g_s}\right).
\end{eqnarray}

\subsection{Hirota Equations}

One may apply standard methods of orthogonal polynomial theory \cite{Szego} to deduce the
existence of a semi-infinite lower unipotent matrix $A$ such that
\begin{eqnarray*}
\mathcal{L} &=& A^{-1}\epsilon A
\end{eqnarray*}
where
\begin{eqnarray*}
\epsilon &=& \begin{pmatrix} 0 & 1 &  \\
                              0 & 0 & 1 \\
			          & 0 & 0  & \ddots \\
                                  &      & \ddots & \ddots 
\end{pmatrix}\,.
\end{eqnarray*}
(For a description of the construction of such a unipotent matrix we refer to Proposition 1 of \cite{EM01}.) 

This is related to the Hankel matrix 
\begin{eqnarray*}
\mathcal{H} &=&  \begin{pmatrix} m_{0} & m_1 &  m_2 & \dots\\
                              m_1 & m_2 & m_3 & \dots \\
			      m_2    & m_3 & m_4  & \dots \\
                                \vdots &  \vdots   & \vdots & \ddots 
\end{pmatrix},
\end{eqnarray*}
where
\begin{eqnarray*}
m_k &=& \int_\mathbb{R} \lambda^k e^{-g_s^{-1}V(\lambda)} d \lambda
\end{eqnarray*}
is the $k^{th}$ moment of the measure, by
\begin{eqnarray*}
A D A^{\dagger} &=&  \mathcal{H}\\
D &=& \mbox{diag}\,\left\{ d_{0}, d_{1} \dots \right\}
\end{eqnarray*}
with
\begin{eqnarray*}
d_n &=& \frac{\det \mathcal{H}_{n+1}}{\det \mathcal{H}_{n}}
\end{eqnarray*}
where $\mathcal{H}_n$ denotes the $n \times n$ principal sub-matrix of $\mathcal{H}$ whose determinant may be expressed as (see Szeg\"o's classical text \cite{Szego}),
\begin{eqnarray}
\nonumber \det \mathcal{H}_n &=& n! \hat{Z}^{(n)} \left(t_1,  t_{2\nu} \right)\\
\label{szego} \hat{Z}^{(n)} \left(t_1, t_{2\nu}\right) &=& \int_\mathbb{R} \cdots \int_\mathbb{R} \exp\left\{
-g_s^{-2}\left[g_s\sum_{m=1}^{n} V(\lambda_{m}; t_1, \ t_{2\nu})  - \right.\right.
\nonumber 
\\
&& \phantom{\int_\mathbb{R} \cdots \int_\mathbb{R} \exp}\hspace{2cm}
\left. \left.
g_s^2 \sum_{m\neq \ell} \log{|
\lambda_{m} -
\lambda_{\ell} | } \right] \right\}  d^{n} \lambda,
\end{eqnarray}
where $V(\lambda; t_1, \ t_{2\nu + 1}) = \frac12 \lambda^2 + t_1 \lambda +  \frac{t_{2\nu}}{2\nu} \lambda^{2\nu}$.
We set $\det \mathcal{H}_0 = 1$.
\medskip

\begin{rem} We sometimes need to extend the domain of the tau functions to include other parameters, such as $t_1$, as we have done here. Doing this presents no difficulties in the prior constructions. 
\end{rem}\smallskip

The diagonal elements may in fact be expressed as
\begin{eqnarray*}
d_n &=& \frac{\tau^2_{n+1, g_s}}{\tau^2_{n, g_s}} d_n(0)
\end{eqnarray*}
where 
\begin{eqnarray}
\label{szego1}\tau^2_{n, g_s} &=& \frac{\hat{Z}^{(n)}\left(t_1, t_{2\nu}\right)}{\hat{Z}^{(n)}\left(0,0\right)}\\ 
\label{szego2} &=& \frac{{Z}^{(n)}\left(t_1, t_{2\nu}\right)}{{Z}^{(n)} \left(0,0\right)}
\end{eqnarray}
which agrees with the definition of the tau function given in (\ref{tausquare}). The second equality follows by reducing the unitarily invariant  matrix integrals in (\ref{szego2}) to their diagonalizations which yields (\ref{szego1}) \cite{EM03}. 
Tracing through these connections, from $\mathcal{L}$ to $D$, one may derive the fundamental identity relating the random matrix partition function to the recurrence coefficients,
\begin{eqnarray}\label{Hirota}
b^2_{n,g_s} = \frac{d_n}{d_{n-1}}&=& \frac{\tau^2_{n+1, g_s}\tau^2_{n-1, g_s}}{\tau^4_{n, g_s}} b^2_{n,g_s}(0)
\end{eqnarray}
which is the basis for our analysis of continuum limits in the next section. (Note that $b^2_{0,g_s}(0) = 0$ and therefore $b^2_{0,g_s} \equiv 0$.) We will also need a differential version of this relation:

\begin{lem} (Hirota) 
\begin{align}
\label{a} a_{n, g_s} &= -g_s \frac{\partial}{\partial t_1} \log \left[ \frac{\tau^2_{n+1,g_s}}{\tau^2_{n, g_s}} \right]  
                                  =  -g_s \frac{\partial}{\partial t_1} \log \left[ \frac{Z^{(n+1)}(t_1, t_{2\nu})}{Z^{(n)}(t_1, t_{2\nu})} \right]\\
\label{b}  b_{n, g_s}^2 &= g_s^2 \frac{\partial^2}{\partial t_1^2} \log \tau^2_{n, g_s}
                                       = g_s^2 \frac{\partial^2}{\partial t_1^2} \log Z^{(n)}(t_1, t_{2\nu})\,,
\end{align}
\end{lem}

\noindent (A derivation of this lemma may be found in \cite{BI}.) It follows from (\ref{b}) and (\ref{I.002}) that
\begin{cor} \label{two-leg}
\begin{eqnarray}\label{b-asymp}  
b_{n, g_s}^2(x; t_{2\nu}) &=& x^2\left( z_0(x; t_{2\nu}) + \dots + \frac1{n^{2g}}z_{g}(x; t_{2\nu}) + \dots\right)\\
z_g(x; t_{2\nu}) &=& \frac{d^2}{dt_1^2} e_g(x; t_1, t_{2\nu})|_{t_1 = 0}
\end{eqnarray}
is a uniformly valid asymptotic asymptotic expansion in the sense of (\ref{I.002} iii).
\end{cor}

\subsubsection{Path Weights and Recurrence Coefficients}  In order to effectively utilize the relations (\ref{fl1}, \ref{fl2}) it will be essential to keep track of how the matrix entries of powers of the recurrence operator, $\mathcal{L}^j$, depend on the original recurrence coefficients. That is best done via the combinatorics of weighted walks on the index lattice of the orthogonal polynomials. For the case of even potentials, the relevant walks are {\it Dyck paths} which are walks, $P$, on $\mathbb{Z}$ which, at each step, can either increase by 1 or decrease by 1. Set 

\begin{equation}  \label{motzkin}
\mathcal{P}^j(m_1, m_2) = \,\, \mbox{the set of all Dyck paths of length $j$ from $m_1$ to $m_2$}.
\end{equation}
Then step weights, path weights and the $(m_1, m_2)$-entry of  $\mathcal{L}^j$ are, respectively, given by 
\begin{eqnarray} \nonumber
\omega(p) &=& \left\{ 
\begin{array}{cc}
 1 & \mbox{if the $p^{th}$ step moves from $n$ to $n+1$ on the lattice}\\
 b^2_n & \mbox{if the $p^{th}$ step moves from $n$ to $n-1$}  
\end{array}
\right.\\
\nonumber \omega(P) &=& \prod_{\mbox{steps}\,\, p \in P} \omega(p)\\
\label{weights} \mathcal{L}^{j}_{m_1, m_2} &=& \sum_{P \in \mathcal{P}^j(m_1, m_2)} \omega(P).
\end{eqnarray}

\subsection{Dyck Representation of the Difference String equations} \label{motzstring}

The {\em difference string equations}  are given (for the $2\nu$-valent case) by  (\ref{fl1}):
\begin{equation}
\left[ \mathcal{L}, \left(\mathcal{L} +  t \mathcal{L}^{2\nu-1} \right)_- \right] = g_s I \,.
\end{equation}
By parity considerations, when the potential $V$ is even, the only non-tautological equations come from the diagonal entries of (\ref{string-star2}):  {\em the $(n, n)$ entry} gives
\begin{equation} \label{string-star2}
g_s = \left(\mathcal{L} +  t \mathcal{L}^{2\nu-1} \right)_{n+1, n} - 
\left(\mathcal{L} +  t \mathcal{L}^{2\nu -1} \right)_{n, n-1} \,.
\end{equation}

\noindent In terms of Dyck paths this becomes
\begin{eqnarray} \label{diff-string}
\frac{x}{n} &=& b^2_{n+1} - b^2_n +  t  \sum_{P \in \mathcal{P}^{2\nu-1}(1,0)}\left(\prod_{m=1}^{\nu}b^2_{n+\ell_m(P)+1} - \prod_{m=1}^{\nu}b^2_{n+\ell_m(P)}\right)
\end{eqnarray}
where $\ell_m(P)$ denotes the lattice location of the path $P$ after the $m^{th}$ downstep and        we have used the relation $x = n g_s$ on the left hand side of the equation.

We illustrate this more concretely for the case of $j = 2\nu =4$.  
Referring to (\ref{motzkin}), the relevant path classes here are
\begin{eqnarray*}
\mathcal{P}^1(1, 0) &=& \mbox{a descent by one step }\\
\mathcal{P}^3(1, 0) &=& \mbox{paths with exactly one upstep and two downsteps (Fig. \ref{0h})}
\end{eqnarray*}

\begin{figure}[h] 
\begin{center}
\resizebox{5in}{!}{\includegraphics{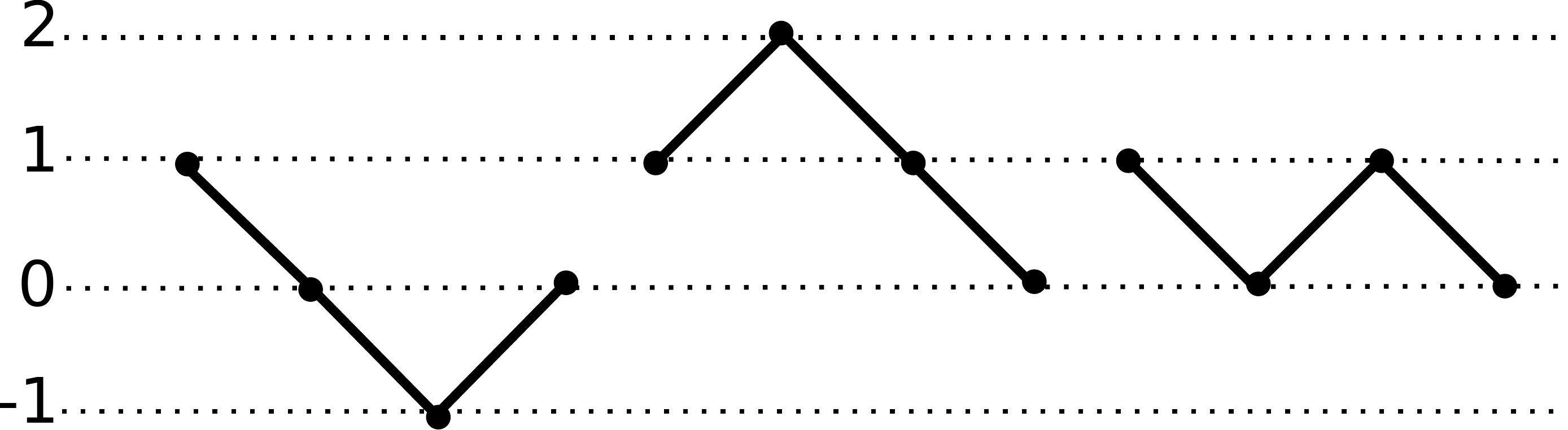}}
\end{center}
\caption{\label{0h} Elements of $\mathcal{P}^3(1, 0)$}
\end{figure}

Note that the structure of the path classes does not actually depend upon $n$. This is a reflection of the underlying spatial homogeneity of these equations. Thus, for the purpose of describing the path classes, one can translate $n$ to $0$. 

Now applying (\ref{weights}) the difference string equation becomes, for $n > 0$,
\begin{eqnarray*}
\frac1n &=& \left( b^2_{n+1} - b^2_n\right) +  t \left\{b^2_{n+1}\left(b^2_{n} + b^2_{n+1} + 
b^2_{n+2}\right) - b^2_n \left(b^2_{n-1} + b^2_{n} + b^2_{n+1}\right) \right\}
\end{eqnarray*}
where, for this example, we have set the parameter $x$ equal to $1$. 

\subsection{Dyck Representation of the Toda equations} \label{motztoda}

We now pass to a more explicit form of of the {\it Toda equations} (\ref{fl2}) in the case $j = 2\nu$:
\begin{eqnarray*}
\label{bn} 
-\frac{x}{n}  \frac{d b^2_n}{dt_{2\nu}} &=&  
\left(\mathcal{L}^{2\nu}\right)_{n+1, n-1} - \left( \mathcal{L}^{2\nu}\right)_{n, n-2}\\
&=& \sum_{P \in \mathcal{P}^{2\nu}(2,0)}\left(\prod_{m=1}^{\nu+1}b^2_{n+\ell_m(P)+1} - \prod_{m=1}^{\nu+1}b^2_{n+\ell_m(P)}\right).
\end{eqnarray*}
Once agian we illustrate these equations in the tetravalent case ($\nu = 2$). The relevant path class is:
\begin{eqnarray*}
\mathcal{P}^4(2, 0) &=& \mbox{paths with exactly one upstep and three downsteps}.
\end{eqnarray*} 
%
Applying (\ref{weights}),  the tetravalent Toda equations become
\begin{align}
- \frac1n \frac{db^2_n}{dt} =
& \, b^2_{n+1} b^2_n \left( b^2_{n - 1} + b^2_{n} + b^2_{n+1} + b^2_{n+2}\right) 
-  b^2_n b^2_{n-1}\left(b^2_{n - 2} + b^2_{n - 1} + b^2_{n} + b^2_{n+1}\right) 
\end{align}
where we have again used the relation $x = n g_s$ and then set the parameter $x = 1$.

\section{Continuum Limits} \label{sec:44}
The continuum limits of the difference string and Toda equations will be described in terms of certain scalings of the independent variables, both discrete and continuous. As indicated at the outset, the positive parameter $g_s$ sets the scale for the potential in the random matrix partition function and is taken to be small. The discrete variable $n$ labels the lattice {\it position} on $\mathbb{Z}^{\geq 0}$ that marks, for instance, the $n^{th}$ orthogonal polynomial and recurrence coefficients.  We also always take $n$ to be large and in fact to be of the same order as $\frac1{g_s}$; i.e., as $n$ and $g_s$ tend to $\infty$ and $0$ respectively, they do so in such a way that their product
\begin{equation}\label{xdef} 
x \doteq   g_s\, n
\end{equation}
remains fixed at a value close to $1$. 

In addition to the {\it global} or {\it absolute} lattice variable $n$, we also introduce a {\it local} or {\it relative} lattice variable denoted by $k$. It varies over integers but will always be taken to be small in comparison to $n$ and independent of $n$. The Dyck lattice paths naturally  introduce the composite discrete variable $n + k$ into the formulation of the difference string and Toda equations which we think of as a small discrete variation around a large value of  $n$. The spatial homogeneity of those equations manifests itself in their all having the same form, independent of what $n$ is, while $k$ in those equations varies over $\{-\nu - 1, \dots, -1,0,1, \dots, \nu + 1\}$, the \emph{bandwidth} of the $(2\nu)^{th}$ Toda/Difference String equations. 
Taking $\nu + 1 << n$ will insure the necessary separation of scales between $k$ and $n$. We
define 
\begin{eqnarray}
\label{wdef}{w} &\doteq& (n+k) g_s\\
&=& x +  g_s\, k = x\left( 1 + \frac kn\right).
\end{eqnarray} 
as a {\em spatial} variation close to $x$ which will serve as a continuous analogue of the lattice location along a Dyck path relative to the starting location of the path.  

We also introduce the self-similar scalings:
\begin{eqnarray}
\label{nuscaling1} s_1 &\doteq& x^{-\frac12} t_1\\
\label{nuscaling2} s_{2\nu} &\doteq& x^{\nu - 1} \frac{t_{2\nu}}{2\nu}
\end{eqnarray}
that are natural given (\ref{sss}).
In terms of these scalings, (\ref{b-asymp}) may be rewritten \cite{EMP08} as 
\begin{eqnarray}\label{bs-asymp}  
b_{n, g_s}^2(s_{2\nu}) &=& x\left( z_0(s_{2\nu}) + \dots + \frac1{n^{2g}}z_{g}(s_{2\nu}) + \dots\right)\\
z_g(s_{2\nu}) &=& \frac{d^2}{ds_1^2} e_g(s_1, s_{2\nu})|_{s_1 = 0}\\
\end{eqnarray}
and, by \cite{Er09} App. A,
\begin{eqnarray} \label{b-shift}
b_{n+k,g_s}^2(s_{2\nu}) &=& x \left(f_0(s_{2\nu}, w) + \dots + \frac1{n^{2g}}f_{g}(s_{2\nu}, w) +
\dots \right)\\
\label{b-shift_g} f_{g}(s_{2\nu}, w) &=& w^{1-2g} z_g(s_{2\nu}w^{\nu-1}).
\end{eqnarray}
\begin{rem}  We mention here that the variables $s_j$ as defined above differ slightly from their usage in related works \cite{EMP08, Er09} where $s_j = -\alpha_j t_j$ for appropriate parameters $\alpha_j > 0$. 
\end{rem}

We also introduce a shorthand notation to denote the expansion of the coefficients of $f(s_1, s_{2\nu+1},{w})$ around $w = x$. 
\begin{defn} \label{shorthand} For ${w} = x + g_s\, k$,
\begin{eqnarray}
\label{f1k} f(s_1, s_{2\nu},w) &=& \sum_{j=0}^\infty \frac{f_{w^{(j)}}|_{w=x}}{j!} \left(\frac{kx}{n}\right)^j
\end{eqnarray}
where the subscript $w^{(j)}$ denotes the operation of taking the $j^{th}$ derivative with respect to $w$ of each coefficient of $f$:
\begin{eqnarray*}
f_{w^{(j)}}   &=& \sum_{g \geq 0} \frac{\partial^{j}}{\partial w^j}  f_g(s_1, s_{2\nu},w) \frac{1}{n^{2g}}
\end{eqnarray*}
As valid asymptotic expansions these representations denote the asymptotic series whose successive terms are gotten by collecting all terms with a common power of $1/n$ in (\ref{f1k}). 
\end{defn}
\medskip

In what follows we will frequently abuse notation and drop the evaluation at $w=x$. In particular, we will write 
\begin{eqnarray}
\label{f1kform} b^2_{n+k, g_s} &=& \sum_{j=0}^\infty \frac{f_{w^{(j)}}}{j!} \left(g_s\, k\right)^j = \sum_{j=0}^\infty \frac{1}{j!} \sum_{g \geq 0} \frac{\partial^{j}}{\partial w^j}  f_g( s_{2\nu},w) \frac{1}{n^{2g}}\left(g_s\, k\right)^j
\end{eqnarray}
In doing this these series must now be regarded as formal but whose orders are still defined by collecting all terms in $1/n$ and $g_s$ of a common order. (Recall that $g_s \sim \frac1n$ so that $n^{-\alpha} g_s^{\beta} = \mathcal{O}(n^{-(\alpha + \beta)})$).  They will be substituted into the difference string and the Toda equations to derive the respective continuum equations. At any point in this process, if one evaluates these expressions at $w=x$ and $g_s = \frac xn$ one may recover valid asymptotic expansions in which the $a_{n+k, g_s}$ and $b^2_{n+k, g_s}$ have their original significance as valid asymptotic expansions of the recursion coefficients. 

\subsection{The Continuum Limit of the Toda equations}
We are now in a position to study the Toda lattice equations (\ref{bn}) expanded on the formal asymptotic series (\ref{f1kform}):
\begin{eqnarray*}
- \frac{1}{n} \frac{d}{ds}f(s,w) &=& \sum_{P \in \mathcal{P}^{2\nu}(1,-1)}\left(\prod_{m=1}^{\nu+1} \sum_{j=0}^\infty \frac{f_{w^{(j)}}}{j!} \left(g_s\, (\ell_m(P) + 1)\right)^j - \prod_{m=1}^{\nu+1}\sum_{j=0}^\infty \frac{f_{w^{(j)}}}{j!} \left(g_s\, \ell_m(P)\right)^j\right).
\end{eqnarray*}
From now on we will take $s_1 = 0$,  since its role in determining the structure of the asymptotic expansions of the $b_{n+k}$ is now completed, and set $s_{2\nu} = s$.  

Collecting terms in these equations order by order in orders of $1/n$ we will have a hierarchy of equations that, in principle, allows one to recursively determine the coefficients of (\ref{bs-asymp}). We will refer to this hierarchy as the {\it Continuum Toda equations}. (Note that one has such a hierarchy for each value of $\nu$.) Of course this is a standard procedure in perturbation theory. The equations we will derive are pdes in the form of evolution equations in which $w$, now regarded as a continuous variable, is the independent {\em spatial} variable and $s_{2\nu}$ is the {\em temporal} variable. One must still determine, at each level of the hierarchy, which solution of the ode is the one that corresponds to the expressions given for $f_g$ in  (\ref{b-shift_g}). This amounts to a kind of solvability condition. This process was carried out fully in \cite{EMP08} and \cite{Er09}. We will now state the results of that analysis.
\begin{thm} \label{cont} \cite{EMP08} The continuum limit, to all orders, of the Toda Lattice equations as $n \to \infty$ is given by the following infinite order partial differential equation for $f(s,w)$:
\begin{align}
\nonumber \frac{df}{ds} & =  F^{(\nu)}\left(g_s; f, f_w, \dots, f_{w^{(j)}}, \dots \right)   \\
\nonumber  
& \doteq \sum_{g \geq 0} g_s^{2g} F_g^{(\nu)}(f, f_w, f_{w^{(2)}}, \cdots,
f_{w^{(2g+1)}}) \\
\label{Burgers} & = c_\nu f^\nu f_w  +
g_s^2 F^{(\nu)}_1(f,f_w,f_{ww},f_{www}) +\cdots\\
& \textrm{for}\,\, (s,w) \,\, \textrm{near}\,\, (0,1)\,\,  \mbox{and initial data given by}\,\, f(0,w) = w.  \\
\label{contHO}
F_g^{(\nu)} &  =  \sum_{\lambda:|\lambda|= 2g+1 \ni \; \ell(\lambda) \leq \nu+1} \frac{d_\lambda^{(\nu,g)}}{\prod_j r_j(\lambda)! } f^{\nu -
\ell(\lambda)+1}\prod_{j} \left( \frac{f_{w^{(j)}}}{j!}
\right)^{r_j(\lambda)} 
\end{align}
where $\lambda= (\lambda_1, \lambda_2, \dots)$ is a
partition, with $\lambda_1 \geq \lambda_2 \geq \lambda_3 \geq \cdots$,  of $2g+1$; $r_j(\lambda) = \# \{\lambda_i | \lambda_i = j\}$; $\ell(\lambda) = \sum_j r_j(\lambda)$ is the {\em length} of $\lambda$; and $|\lambda| = \sum_i \lambda_i$ is the {\em size} of $\lambda$; and 
$d_\lambda^{(\nu,g)}$ are coefficients to be described in the next proposition. By (\ref{Burgers}), $d_\Box^{(\nu,0)} = c_\nu$.
\end{thm}

The above result effectively reduces the determination of the hierarchy to an enumeration in terms of a pair of partitions. The first class of partitions is fairly straightforward and amounts to keeping track of the partial derivatives that enter into the expressions at a given level. Thus at $\mathcal{O}(g_s^{2g+1} )$ the terms are products of partial derivatives of various orders (the {\it parts}) which must add up to $2g+1$. These terms correspond to the tableaux at the $(2g+1)^{st}$ level of the {\it Hasse-Young graph} shown in the left panel of Figure \ref{ZigZag}. (In this we ignore, for now, the powers of $1/n$ that are internal to the asymptotic series $f$ and its $w$-derivatives.) The other type of partition relates to how the Dyck paths enter our equations. We have already seen that a Dyck path is completely determined by specifying when its downsteps occur. The Toda equations depend explicitly on where  these downsteps (the $\ell_m(P)$ appearing in the equations as stated at the start of this subsection). However these two specifications can be related and the downstep times are encoded in terms of a partition that measures deviation of the path from a standard zig-zag path. The following proposition gives an explicit closed form expression for the coefficients $d_\lambda^{(\nu,g)}$ in terms of both classes of partitions.

\begin{prop} \cite{Er09} ({\em App. A.3})\label{d-coeffs}
\begin{eqnarray*}
d_\lambda^{(\nu,g)} &=&  \sum_{\begin{array}{c}(\nu+1, \nu, \dots, 2,1) \subseteq \mu \subseteq (2\nu, 2\nu-1, \dots, \nu)\\ \mu \in \mathcal{R}\end{array}} 2 \,\, m_\lambda \left(\mu_1- \eta_1,  \dots, \mu_{\nu+1} - \eta_{\nu+1}\right)
\end{eqnarray*}
where $\mathcal{R}$ is the set of {\it restricted partitions} (meaning that $\mu_1 > \mu_2 > \cdots > \mu_{\nu+1} $), $(\eta_1, \dots, \eta_{\nu+1}) = (2\nu, 2\nu-2, \dots, 2,0)$, and $m_\lambda(x_1, \dots, x_{\nu+1})$ is the monomial symmetric polynomial associated to $\lambda$ \cite{Mac}. The relation of inclusion between partitions, $\rho \subseteq \mu$ means that $\mu_j \geq \rho_j$ for all $j$.
\end{prop}

The right panel of Figure \ref{ZigZag} exemplifies, in the case when $\nu = 4$, a geometric realization of the partitions $\mu$ being summed over in the above formula. Such a partition corresponds to a zig-zag path contained in the $\nu \times \nu + 1$ rectangle which starts at the leftmost corner with the $0^{th}$ step being a  up-step and terminates at the rightmost corner with the $(2\nu)^{th}$ step. The red path corresponds to the distinguished partition $\eta$ which records the downsteps at times $2\nu, 2\nu - 2, \dots, 0$ so that in this case the initial ($0^{th}$) step is a downstep. The green path illustrates a typical $\mu$ which, in this case, takes downsteps at times $(7,6,5,3,2)$.   Given such a partition-path, one may project the initial point of each of its downsteps to the horizontal axis as indicated in the figure for the red and green paths. Then $\mu_i - \eta_i$ equals the signed separation between the $i^{th}$ green point and the $i^{th}$ red point, reading from right to left. The symmetric polynomial $m_\lambda$ then gets evaluated at these separation values in the above formula for $d_\lambda^{(\nu,g)}$. This description could also have been formulated in terms of the {\it border strips} associated to the {\it skew tableaux} $\mu/(\nu+1, \nu, \dots, 2,1)$ but we will not elaborate on that here. 

\begin{figure} [h]  
   \centering
   \includegraphics[width=2in]{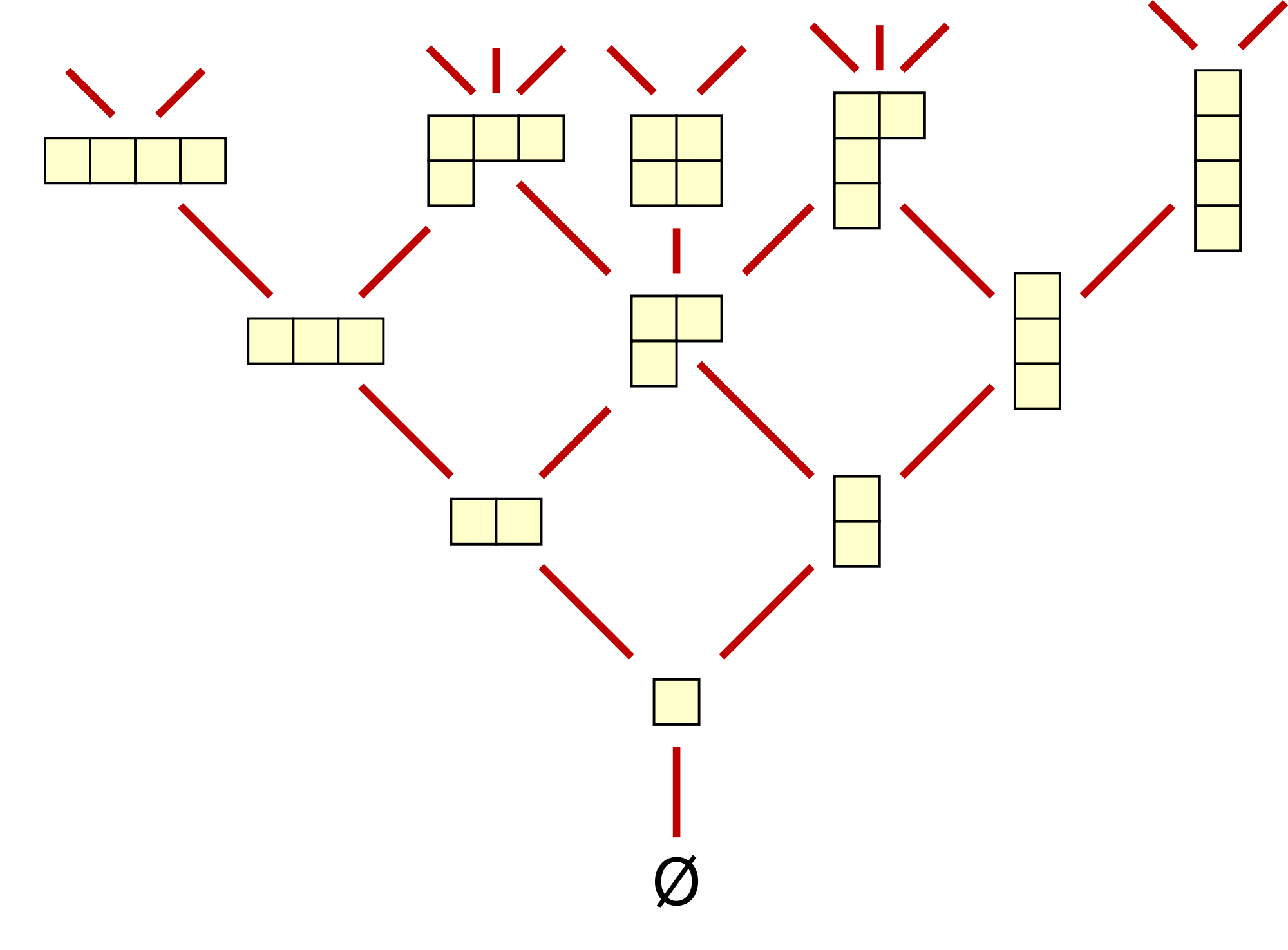} 
   \includegraphics[width=2in]{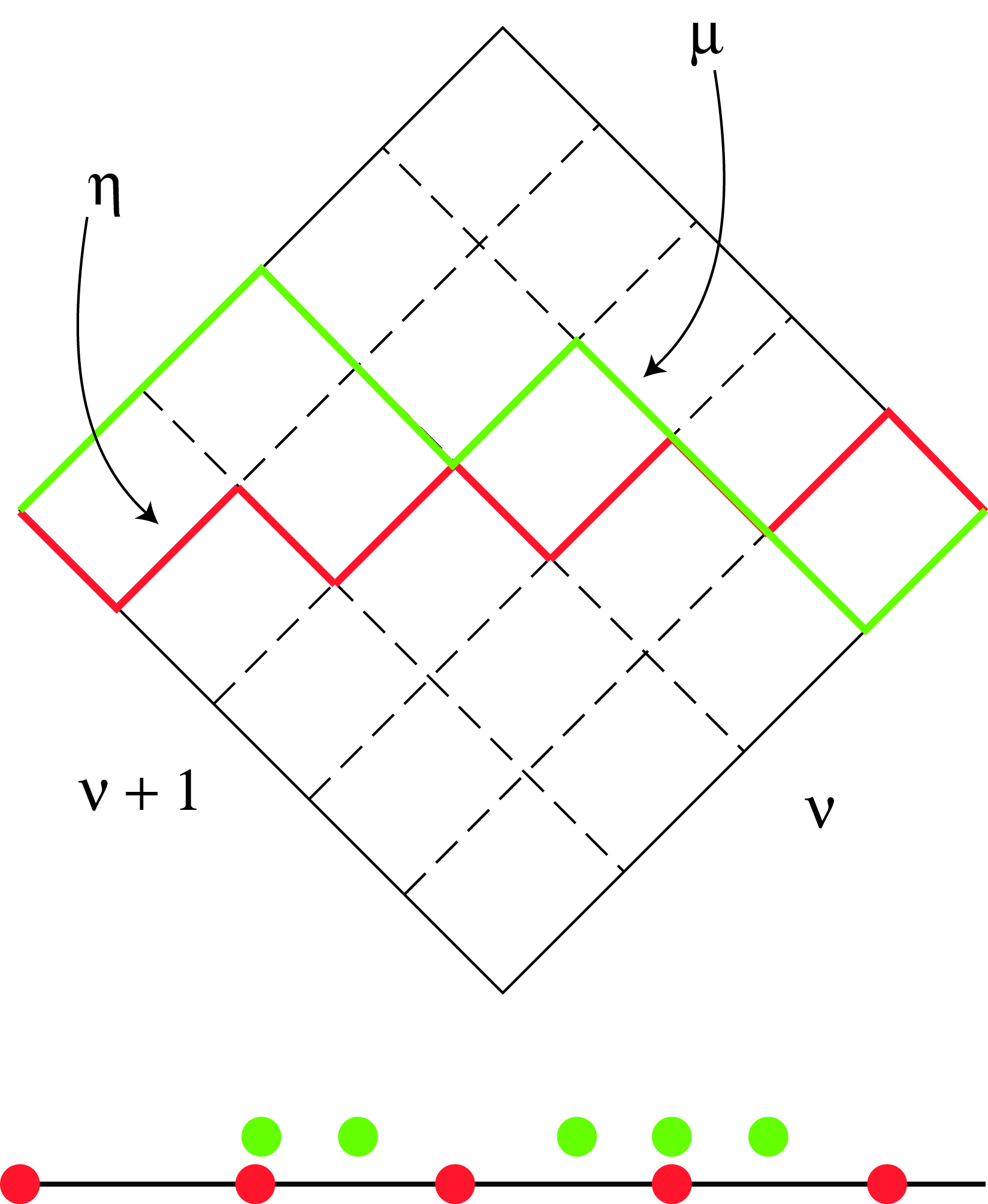}
   \caption{left: $\lambda$ Hasse-Young graph (courtesy D. Eppstein);  right: $\mu$ partition walks}
   \label{ZigZag}
   \end{figure}

One is now in a position to deduce the form of the Toda hierarchy. This is done by setting $x=1$ so that $g_s = 1/n$. One then collects {\em all} terms of order $n^{-2g}$ in the resulting expansion of (\ref{Burgers}) and this will be a partial differential equation in $s$ and $w$ that we refer to as the $g^{th}$ equation in the continuum Toda hierarchy.

At leading order in the hierarchy one observes that, for general $\nu$, the continuum Toda equation is an inviscid Burgers equation \cite{Wh}
\begin{eqnarray} \label{burgers}
\frac{d}{ds}f_0 &=& - \frac{c_\nu}{\nu+1} \partial_w (f_0)^{\nu + 1}
\end{eqnarray}
with initial data $f_0 = w$. A solution exists and is unique for sufficiently small values of $s$. It may be explicitly calculated by the method of characteristics, also known as the {\it hodograph} method in the version we now present.  Consider the (hodograph) relation among the independent variables $(s,w,f_0)$,
\begin{eqnarray}\label{hodograph}
w &=& f_0 + c_\nu s f_0^{\nu}.
\end{eqnarray}

\begin{lem} \label{hodlemma}
A local solution of (\ref{burgers}) is implicitly defined by (\ref{hodograph}).
\end{lem}
\begin{proof}
The annihilator of the differential of (\ref{hodograph})
\begin{eqnarray*}
\left( 1 + \nu c_\nu s  f_0^{\nu-1}\right) df_0 &-& dw + c_\nu  f_0^{\nu} ds  
\end{eqnarray*}
is a two-dimensional distribution locally on the space $(s,w,f_0)$. An initial curve over the $w$-axis (parametrized as the graph of a function $f_0(w)$), transverse to the locus where 
$1 + \nu c_\nu s  f_0^{\nu-1} = 0$ locally determines a unique integral surface foliated by the integral curves of the vector field of the {\it characteristic} vector field
\begin{eqnarray}\label{char1}
\frac{df_0}{ds} &=& 0\\
\label{char2}\frac{dw}{ds} &=& c_\nu f_0^\nu\\
\label{char3}f_0(0,w) &=& f_0(w).
\end{eqnarray}
Equation (\ref{char1}) requires that along an integral curve of the characteristic vector field, $f_0$ is constant; i.e.,
\begin{eqnarray*}
0 &=& \frac{df_0}{ds}(s,w(s)) = \frac{\partial f_0}{\partial s} + \frac{\partial f_0}{\partial w} \frac{dw}{ds}\\
&=& \frac{\partial f_0}{\partial s} + c_\nu f_0^\nu\frac{\partial f_0}{\partial w}
\end{eqnarray*}
by (\ref{char2}) which is equivalent to (\ref{burgers}). Using (\ref{char3}) to set $f_0(0,w) = w$ pins down our solution uniquely.
\end{proof}
\begin{rem}
We note that the numerical coefficients appearing in these Burgers equations depend only on the total number of Dyck paths in $\mathcal{P}^{2\nu}(2,0)$. 
\end{rem}

One finds \cite{EMP08}  from (\ref{hodograph}) and the self similar form of $f_0$,
\begin{eqnarray} \label{z_0}
f_0(s, w) &=& wz_0(s w^{\nu - 1}), \, \mbox{that}\\
z_0(s) &=& \sum_{j \geq 0} c_\nu^j \zeta_j s^j \,\,\, \mbox{where}\\
 c_\nu &=& 2\nu {2\nu -1 \choose \nu - 1} = (\nu + 1) {2\nu \choose \nu + 1}
\,\,\, \mbox{and}\\ \label{catalan}
 \zeta_j &=&   \frac{1}{j} {\nu j \choose j - 1} = \frac{1}{(\nu-1)j+1} {\nu j \choose j}.
\end{eqnarray}
When $\nu = 2$, $\zeta_j$ is the $j^{th}$ Catalan number. For general $\nu$ these are the {\it higher Catalan numbers} which play a role in a wide variety of enumerative combinatorial problems.

\subsection{Continuum Limits of the Difference String Equations} \label{ssec:dse}
The {\em Continuum Difference String} hierarchies may be derived from the difference string equations (\ref{diff-string}) in a manner completely analogous to what was done with the Toda equations in the previous subsection.  
\medskip

Expanding (\ref{diff-string}) on the asymptotic series (\ref{f1kform}) we arrive at the following asymptotic equations.

\begin{eqnarray*}
 \frac{1}{n}  &=& \sum_{j = 1}^\infty \frac{f_{w^{(j)}}}{j!} \left(\frac{1}{n}\right)^j\\
 &+&  2\nu s\sum_{P \in \mathcal{P}^{2\nu - 1}(0,-1)}\left(\prod_{m=1}^{\nu} \sum_{j=0}^\infty \frac{f_{w^{(j)}}}{j!} \left(\frac{\ell_m(P) + 1}{n}\right)^j - \prod_{m=1}^{\nu}\sum_{j=0}^\infty \frac{f_{w^{(j)}}}{j!} \left(\frac{\ell_m(P)}{n}\right)^j\right).
\end{eqnarray*}

The equations at leading order, $\mathcal{O}(n^{-1})$, are
\begin{eqnarray*} \label{leadstring}
1 &=& \partial_w f_0 +  2\nu s \sum_{P \in \mathcal{P}^{2\nu - 1}(1,0)} \nu f_0^{\nu - 1} \partial_w f_0\\
  &=& \partial_w f_0 +  2\nu {2\nu-1 \choose \nu} s  \nu f_0^{\nu - 1} \partial_w f_0
\end{eqnarray*}
or, equivalently,
\begin{eqnarray}\label{leadstringsoln}
\partial_w \left(w - f_0 - c_\nu s f_0^\nu\right) = 0
\end{eqnarray}
which one directly recognizes as the spatial derivative of the hodograph solution (\ref{hodograph}). Evaluating that solution at $w=1$ yields
\begin{eqnarray} \label{stringeqn}
c_\nu s z_0^\nu + z_0 -1 = 0
\end{eqnarray}
which is the functional equation for the generating function of the $\nu^{th}$ higher Catalan numbers, mentioned in the previous subsection.
\medskip

The terms of the equations at  $\mathcal{O}(n^{-2g-1})$ can be computed directly and are found to have the form

\begin{eqnarray} \label{cont-string}
\partial_w\left[f_g + c_\nu s \nu f_0^{\nu - 1} f_g\right] &+& 2s\left(c_\nu
\partial_w \sum_{\begin{array}{c}
  0 \leq k_j < g \\
  k_1 + \dots + k_{\nu} = g\\
\end{array}} f_{k_1}\cdots f_{k_{\nu}}\right)\\
\nonumber  + \partial_w \sum_{k=0}^{g-1} \frac{f_{k w^{(2g-2k)}}}{(2g-2k+1)!}
&+& 2\nu s \left(F_1^{(\nu - 1)}[2g-2] +
F_2^{(\nu - 1)}[2g-4] + \cdots + F_{g}^{(\nu - 1)}[0]\right) = 0, 
\end{eqnarray}
where $F_g^{(\nu - 1)}[2m]$ denotes the coefficient of $n^{-2m}$ in 
\begin{eqnarray} \label{Fg}
F_g^{(\nu - 1)} &=&  \sum_{\lambda:|\lambda|= 2g+1 \ni \; \ell(\lambda) \leq \nu} \frac{\delta_\lambda^{(\nu-1,g)}}{\prod_j r_j(\lambda)! } f^{\nu -
\ell(\lambda)}\prod_{j} \left( \frac{f_{w^{(j)}}}{j!}
\right)^{r_j(\lambda)}\, ;\\
\nonumber \delta_\lambda^{(\nu-1,g)} &=&  \sum_{\begin{array}{c}(\nu, \dots, 2,1) \subseteq \mu \subseteq (2\nu-1,  \dots, \nu)\\ \mu \in \mathcal{R}\end{array}} 2 \,\, m_\lambda \left(\mu_1- \tilde{\eta}_1,  \dots, \mu_{\nu} - \tilde{\eta}_{\nu}\right)\\
\nonumber (\tilde{\eta}_1, \dots, \tilde{\eta}_\nu) &=& (2\nu-1, 2\nu-3, \dots, 1)
\end{eqnarray}

In fact this equation is also $w$-exact! To see this we introduce indeterminates $x_i$ and define
\begin{eqnarray}
f^{(i)}_{w^{(j)}}(s, w) &=& f_{w^{(j)}}(s, x_i w). 
\end{eqnarray} 
We also recall that the monomial symmetric function associated to the partition $\lambda$ is given by
\begin{eqnarray}
m_\lambda(x_1, \dots , x_\nu) &=& x_1^{\lambda_1} \cdots x_\nu^{\lambda_\nu} + \dots
\end{eqnarray} 
where $+ \dots$ on the right hand side denotes adding the (minimal) number of terms which are gotten by permuting the variables $x_i$ in the first term such that the resulting polynomial is symmetric with respect to $S_\nu$.  Finally, define
\begin{eqnarray} \label{potential}
&&\widehat{F}_g^{(\nu - 1)} = \frac{2}{2g+1}\\ 
\nonumber &\times&  \sum_{\begin{array}{c}(\nu, \dots, 2,1) \subseteq \mu \subseteq (2\nu-1, \dots, \nu-1)\\ \mu \in \mathcal{R}\end{array}}\sum_{\begin{array}{c}|{\hat{\lambda}|}= 2g\\ \ell(\hat{\lambda}) \leq \nu \end{array}}\frac1{\prod_j r_j(\hat{\lambda})! }\sum_{\sigma \in S_\nu}\left(\prod_{i=1}^\nu  \frac{f^{(\sigma(i))}_{w^{(\hat{\lambda}_{\sigma(i)})}}}{\hat{\lambda}_{\sigma(i)}!} \right) \Big|_{\vec{x} = \mu - \eta}.
\end{eqnarray}
The following proposition shows that the $\mathcal{O}\left(n^{2k-2g}\right)$ terms of $\hat{F}_k^{(\nu - 1)}$ are an anti-derivative for the terms $F_k^{(\nu - 1)}[2g-2k]$ appearing in (\ref{cont-string}).

\begin{prop} \label{higher_exact}
\begin{eqnarray*}
F_g^{(\nu - 1)} &=& \partial_w \widehat{F}_g^{(\nu - 1)}
\end{eqnarray*}
\end{prop}
\begin{proof}
By symmetry under the action of the symmetric group, it will suffice to check this identity for the single monomial term corresponding to  $x_1^{\lambda_1} \cdots, x_\nu^{\lambda_\nu}$. We directly calculate
\begin{eqnarray*}
\partial_w \prod_{i=1}^\nu  \frac{f^{(i)}_{w^{(\hat{\lambda}_i)}}}{\hat{\lambda}_i!} &=& \sum_{i^\prime=1}^\nu \left(\hat{\lambda}_{i^\prime} + 1\right) \left(\prod_{i \ne i^\prime} 
\frac{f^{(i)}_{w^{(\hat{\lambda}_i)}}}{\hat{\lambda}_i!} \right)
\frac{f^{(i^\prime)}_{ w^{(\hat{\lambda}_{i^\prime}+1)}}}{\left(\hat{\lambda}_{i^\prime} + 1\right)!}\\
&=& \sum_{i^\prime=1}^\nu \left(\hat{\lambda}_{i^\prime} + 1\right) \left(\prod_{i \ne i^\prime} 
x_i^{\hat{\lambda}_i}\frac{f_{w^{(\hat{\lambda}_i)}}}{\hat{\lambda}_i!} \right)
x_{i^\prime}^{\hat{\lambda}_{i^{\prime} +1}}\frac{f_{ w^{(\hat{\lambda}_{i^\prime}+1)}}}{\left(\hat{\lambda}_{i^\prime} + 1\right)!}
\end{eqnarray*}
where in the second line we have made changes of variables replacing $x_i w$ by $w$.
Using this we observe that the inner summands of $\partial_w \widehat{F}_g^{(\nu - 1)}$ have the form
\begin{eqnarray*}
\sum_{\begin{array}{c}|{\hat{\lambda}|}= 2g\\ \ell(\hat{\lambda}) \leq \nu \end{array}}\frac{\partial_w}{\prod_j r_j(\hat{\lambda})! } \prod_{i=1}^\nu  \frac{f^{(i)}_{w^{(\hat{\lambda}_i)}}}{\hat{\lambda}_i!} &=&\sum_{\begin{array}{c}|{\hat{\lambda}|}= 2g\\ \ell(\hat{\lambda}) \leq \nu \end{array}} \frac1{\prod_j r_j(\hat{\lambda})! }\sum_{i^\prime=1}^\nu \left(\hat{\lambda}_{i^\prime} + 1\right) \left(\prod_{i \ne i^\prime} 
\frac{f^{(i)}_{w^{(\hat{\lambda}_i)}}}{\hat{\lambda}_i!} \right)
\frac{f^{(i^\prime)}_{ w^{(\hat{\lambda}_{i^\prime}+1)}}}{\left(\hat{\lambda}_{i^\prime} + 1\right)!}\\
&=& \sum_{\begin{array}{c}|{\lambda}|= 2g+1\\ \ell(\lambda) \leq \nu + 1\end{array}} \frac1{\prod_j r_j(\lambda)! }\sum_{j\geq 1} j r_j(\lambda) x_1^{\lambda_1} \cdots, x_\nu^{\lambda_\nu} \prod_{i=1}^\nu  \frac{f_{w^{(\lambda_i)}}}{\lambda_i!}\\
&=& (2g+1) \sum_{\begin{array}{c}|{\lambda}|= 2g+1\\ \ell(\lambda) \leq \nu + 1\end{array}}  
\frac1{\prod_j r_j(\lambda)! } x_1^{\lambda_1} \cdots, x_\nu^{\lambda_\nu} \prod_{j\geq 1} \left(\frac{f_{w^{(j)}}}{j!}\right)^{r_j(\lambda)}.
\end{eqnarray*}
In the second line the two inner summations of the first line have been rewritten in terms of clusters of partitions adjacent to a given partition of size $2g+1$. To get the last line we use the fact that 
$\sum_{j\geq 1} j r_j(\lambda) = |\lambda| =2g+1$.
 The proposition now follows from this observation, (\ref{potential}) and (\ref{Fg}).
\end{proof}

As a consequence of this result one sees that the continuum difference string equation is directly integrable:
\begin{eqnarray*}
f_g &=& \frac{-2s}{1 +  c_\nu s \nu f_0^{\nu - 1}} \left\{ \left(c_\nu \sum_{\begin{array}{c}
  0 \leq k_j < g \\
  k_1 + \dots + k_{\nu} = g\\
\end{array}} f_{k_1}\cdots f_{k_{\nu}}\right)\right.\\
\nonumber  &+& \left.  \nu  \left(\widehat{F}_1^{(\nu - 1)}[2g-2] +
\widehat{F}_2^{(\nu - 1)}[2g-4] + \cdots + \widehat{F}_{g}^{(\nu - 1)}[0]\right)
+ \frac{1}{2s}\sum_{k=0}^{g-1} \frac{f_{k w^{(2g-2k)}}}{(2g-2k+1)!}\right\}.
\end{eqnarray*} 
Setting $w=1$ and applying (\ref{stringeqn}) to eliminate $s$ this reduces to
\begin{eqnarray*}
z_g &=& \frac{2z_0 (z_0 - 1)}{ \left(\nu - (\nu - 1) z_0\right)} \left\{ \left(\sum_{\begin{array}{c}
  0 \leq k_j < g \\
  k_1 + \dots + k_{\nu} = g\\
\end{array}} \frac{z_{k_1}}{z_0}\cdots \frac{z_{k_{\nu}}}{z_0}\right)
+ \frac{1}{2(1 - z_0)} \sum_{k=0}^{g-1} \frac{f_{k w^{(2g-2k)}}|_{w=1}}{(2g-2k+1)!}\right.\\
\nonumber  &+& \left. \frac{\nu}{c_\nu z_0^\nu} \left(\widehat{F}_1^{(\nu - 1)}[2g-2] +
\widehat{F}_2^{(\nu - 1)}[2g-4] + \cdots + \widehat{F}_{g}^{(\nu - 1)}[0]\right)\Big|_{w = 1} \right\}.
\end{eqnarray*}
It is immediate from this representation that $z_g$ is a rational function of $z_0$. Apriori this {\em anti-derivative} should also include a constant term (in $w$; it could depend on $s$). This would lead to a term of the form $c(s)/\left(\nu - (\nu - 1) z_0\right)$. However, in \cite{Er09} it is shown, by an independent argument, that the pole order in $z_0$ at $\nu/(\nu-1)$ is always greater than one. Hence the constant of integration must be zero. With further effort this can be refined to 
\begin{thm} \cite{Er09}\label{result}
\begin{eqnarray} \label{rational}
 z_g (z_0) &=& \frac{z_0  (z_0 -1) P_{3g-2}(z_0)}{(\nu -(\nu -1)z_0)^{5g-1}}\\
 \nonumber &=& z_0 \left\{ \frac{a_0^{(g)}(\nu)}{(\nu - (\nu-1)z_0)^{2g}} + \frac{a_1^{(g)}(\nu)}{(\nu - (\nu-1)z_0)^{2g+1}}+ \cdots + \frac{a_{3g-1}^{(g)}(\nu)}{(\nu - (\nu-1)z_0)^{5g-1}}\right\}, 
\end{eqnarray}
where $P_{3g-2}$ is a polynomial of degree $3g-2$ in $z_0$ whose coefficients are rational functions of $\nu$ over the rational numbers $\mathbb{Q}$ and $a_{3g-1}^{(g)}(\nu) \ne 0$. 
\end{thm}
\medskip

\begin{rem} \label{diffposet}
A key element in the proof of Proposition \ref{higher_exact} is the observation that differentiation with respect to $w$ adjusts the multinomial labelling of partial derivatives in the expansion according to the edges of the Hasse-Young graph (Fig \ref{ZigZag}). This graph describes the adjacency relations between Young diagrams of differing sizes. The edges describe which partitions of size $2g+1$ are {\it covered} by a given partition of size $2g$. Conversely it describes which partitions of size $2g$ cover a partition of size $2g+1$ which in the setting described here acts as an anti-differentiation operator. This kind of structure was called a differential poset by Stanley and systematically examined in \cite{St}. 
\end{rem}

\section{Determining $e_g$} \label{sec:4}

Recalling the basic identity (\ref{Hirota})
\begin{equation}
b_n^2 = \frac{\tau^2_{n+1} \tau^2_{n-1}}{\tau^4_n} b_n^2(0) \,,
\end{equation}
we have, by taking logarithms, 
\begin{equation} \label{tauk-2nddiff} 
\log \tau^2_{n+1} - 2 \log \tau^2_n + \log \tau^2_{n-1} = \log(b^2_n) - \log(b^2_n)(0)\,,
\end{equation}
where the initial value $b^2_n(0) = x$ is given by the recursion relations  of the Hermite polynomials.  
As in \cite{EMP08}, we can use formula \eqref{tauk-2nddiff} to recursively determine $e_g$ in terms of solutions to the continuum equations.  We use the asymptotic expansion of $b^2_n$ which has the form (\ref{bs-asymp}):
\begin{align} \nonumber
 b_n^2 &= x \sum_{g=0}^\infty z_g(s) n^{-2g}\,,
\end{align}

Note that the left hand side of equation (\ref{tauk-2nddiff}) has the form of a centered second difference, $\Delta_1 \tau^2_{n,n} - \Delta_{-1}\tau^2_{n.n}$. It follows \cite{EMP08} that this expression has an expansion for large $n$ involving only even derivatives of the spatial variable $w$.   We have, at order $n^{-2g}$,  
\begin{align} \label{Hirota2}
\frac{\partial^2}{\partial w^2} E_g(s,w)& =  -\sum_{\ell = 1}^g
\frac{2}{(2\ell + 2)!} \frac{\partial^{2\ell + 2}}{\partial w^{2\ell + 2}} E_{g - \ell}(s,w)\\ 
& \nonumber + \mbox{the}\,\,\,  n^{-2g} \,\,\, \, \mbox{terms of}\,\,\, \log\left( \sum_{m=0}^\infty \frac1{n^{2m}} f_m(s)\right) 
\end{align}
where $E_h(s,w) = w^{2 - 2h} e_h(w^{\nu - 1} s)$. In \cite{Er09} it was shown that $e_g$ is rational in $z_0$ with poles located only at $z_0 = \frac\nu{\nu - 1}$.  However we will now prove the more refined result stated in Theorem \ref{thm51}.  

The proof of this result is by induction on $g$. (The base case of $g = 2$ is established by direct calculation \cite{EMP08}.) We assume that (\ref{note}) holds for all $k < g$. 
We state here, without proof, some straightforward lemmas and propositions describing the derivatives of (\ref{note})  (details may be found in \cite{Er09} where similar lemmas are proved for the $z_g$). Set
\begin{eqnarray*}
E_k(w,s) &=& w^{2-2k} e_k(w^{\nu-1}s)
\end{eqnarray*}
\begin{lemma} \label{lem51}
\begin{eqnarray*}
\left(E_k\right)_{w^{(p)}}(s,w) &=& \sum_{j=0}^p (\nu-1)^j Q_j^{(p,k)}(\nu) w^{2-2k +(\nu-1)j - p} s^j  e_k^{(j)}\\
\mbox{where}\,\,\, e_k^{(j)} &=& d^j e_k/d\tilde{s}^j;\,\, \tilde{s} = sw^{\nu-1}\\
Q_j^{(p,k)}(\nu) &=& Q_{j-1}^{(p-1,k)}(\nu) + \{(\nu-1)j - (2k-3+p)\} Q_j^{(p-1,k)}(\nu) \qquad 0 < j < p\\
Q_0^{(p,k)}(\nu) &=& (2-2k)_p \qquad p > 0\\
Q_p^{(p,k)}(\nu) &=& 1\\
Q_j^{(p,k)}(\nu) &=& 0 \qquad j > p,\,\,\,  j<0. 
\end{eqnarray*}
\end{lemma}
\begin{lemma} \label{lem52} For $1 < k < g$ and $j > 0$,
\begin{eqnarray*}
e_k^{(j)} &=& (-1)^j c_\nu^j z_0^{j\nu+1} \left(\sum_{\ell = 0}^{3k-4+j}\frac{c_\ell^{(k,j)}(\nu)}{(\nu - (\nu-1)z_0)^{2k + \ell+j-1}}\right)\\
c_\ell^{(k,j)}(\nu) &=& \left[(j-1)\nu - (2k+\ell + (j-3))\right] c_\ell^{k,j-1}(\nu) + \nu(2k + \ell +(j-3)) c_{\ell -1 }^{k,j-1}(\nu)\\
c_\ell^{(k,j)}(\nu) &=& 0 \qquad \ell < 0, \qquad \ell \geq 3k-3+j\\
c_\ell^{(k,0)}(\nu) &=& c_\ell^{(k)}(\nu).
\end{eqnarray*}
\end{lemma}
\begin{lemma} \label{lem53} For $1 < k < g$ and $j > 0$
\begin{eqnarray*}
s^j e_k^{(j)} &=& \frac{z_0}{(\nu-1)^j} \sum_{r=0}^j \sum_{m=r}^{3k-4+j+r} \frac{(-1)^{j-r} {j \choose r} c_{m-r}^{(k,j)}}{\left(\nu - (\nu - 1)z_0\right)^{2k+m-1}}
\end{eqnarray*}
where $m=\ell + r$.
\end{lemma}
\begin{prop} \label{prop52} For $1 < k < g$ and $p > 0$
\begin{eqnarray*}
\left(E_k\right)_{w^{(p)}}(s,1) &=&  (2-2k)_p C^{(k)} +  z_0 \sum_{m=0}^{3k+2p-4} \frac{\sum_{j=0}^{p}Q_j^{(p,k)}(\nu) \sum_{r=0}^{m} (-1)^{j-r}{j \choose r} c_{m-r}^{(k,j)}(\nu)}{\left(\nu - (\nu - 1)z_0\right)^{2k+m-1}}.
\end{eqnarray*}
\end{prop}
But in fact, by the following vanishing lemma
\begin{lemma} \label{lem54}
\begin{eqnarray*}
\sum_{j=0}^p Q_j^{(p,k)} (\nu) \sum_{r=0}^m (-1)^{j-r} {j \choose r} c_{m-r}^{(k,j)}(\nu) &=& 0
\end{eqnarray*}
\end{lemma}
for $m = 0, 1, \dots, p-1$, the minimal pole order of the expansion in Proposition \ref{prop52} is $\geq 2k-2+p$. In particular the minimal pole orders coming from terms involving $E_k$ on the right-hand side of (\ref{Hirota2}) are all greater than $2k-2 + (2g +2 -2k) = 2g$. 
\begin{prop} \label{prop53}
The $n^{-2g}$ terms of $\log\left( \sum_{m=0}^\infty \frac1{n^{2m}} f_m(s)\right)\Big|_{w=1}$
\begin{eqnarray*}
&=&  \sum_{|\lambda| = g} \frac{g!}{\prod_{j \geq 1} r_j(\lambda)!} \prod_{j \geq 1} \left(\frac{z_j}{z_0}\right)^{r_j(\lambda)} \\
&=& \sum_{|\lambda| = g} \frac{g!}{\prod_{j \geq 1} r_j(\lambda)!} \prod_{j \geq 1} 
\left\{\frac{a_0^{(j)}(\nu)}{(\nu - (\nu-1)z_0)^{2j}} +  \cdots + \frac{a_{3j-1}^{(j)}(\nu)}{(\nu - (\nu-1)z_0)^{5j-1}} \right\}^{r_j(\lambda)} 
\end{eqnarray*} 
by (\ref{rational}). 
\end{prop}
This result shows that the minimal pole order coming from the $\log$ terms in (\ref{Hirota2}) is once again greater than $\sum_{j\geq 1} 2j r_j = 2g$.

The preceding lemmas and propositions provide explicit Laurent expansions (in $z_0$) for all terms on the right hand side of (\ref{Hirota2}) with two exceptions:
\begin{eqnarray} \label{E1}
\frac{\partial^{2g}}{\partial w^{2g}} E_1(s,w) &=& \frac{\partial^{2g}}{\partial w^{2g}} e_1(w^{\nu - 1}s); \\
\nonumber e_1 &=& \frac{-1}{12} \log (\nu  - (\nu-1)z_0)\\
\label{E0} \frac{\partial^{2g+2}}{\partial w^{2g+2}} E_0(s,w) &=& \frac{\partial^{2g+2}}{\partial w^{2g+2}} 
w^2 e_0(w^{\nu - 1}s); \\
\nonumber e_0 &=& \frac12 \log z_0 + \frac{(\nu-1)^2}{4\nu(\nu+1)}\left(z_0 - 1\right)\left(z_0 - \frac{3(\nu + 1)}{\nu - 1} \right).
\end{eqnarray}
With a small modification (\ref{E1}) may be brought in line with Proposition \ref{prop52},
\begin{prop} \label{prop54}
\begin{eqnarray*}
\left(E_1\right)_{w^{(p)}}(s,1) &=& z_0 \sum_{m=0}^{2p-1} \frac{\sum_{j=1}^{p}Q_j^{(p,1)}(\nu) \sum_{r=0}^{m} (-1)^{j-r}{j \choose r} c_{m-r}^{(1,j)}(\nu)}{\left(\nu - (\nu - 1)z_0\right)^{m+1}}\\
&=& - \frac{(p-1)!}{12} \\
 &+& \frac{1}{\nu - 1} \sum_{m=1}^{2p} \frac{\sum_{j=1}^p Q_j^{(p,1)}(\nu) \sum_{r=0}^m (-1)^{j-r} {j \choose r} \left\{\nu c_{m-r-1}^{(1,j)}(\nu) - c_{m-r}^{(1,j)}(\nu)\right\}}{\left(\nu - (\nu-1)z_0\right)^m}
\end{eqnarray*}
with $c_0^{(1,1)} = \frac{\nu - 1}{12}$. All other coefficients are then specified by the corresponding recursions stated in Lemmas \ref{lem51} - \ref{lem53} with $k$ set to $1$.
\end{prop}

\noindent A variant of the vanishing Lemma \ref{lem54} also holds for $\left(E_1\right)_{w^{(p)}}(s,1)$:
\begin{lemma}
\begin{eqnarray*}
\sum_{j=1}^p Q_j^{(p,1)}(\nu) \sum_{r=0}^m (-1)^{j-r} {j \choose r} \left(c_{m-r}^{(1,j)}(\nu) - \nu c_{m-r-1}^{(1,j)}(\nu)\right) &=& 0
\end{eqnarray*}
\end{lemma}
for $m = 1, \dots, p - 1$. It follows that the minimal pole order of the expansion in Proposition \ref{prop54} is at least $p$ and so the corresponding contribution to the minimal pole order of (\ref{Hirota2}) is $\geq2g$. 

Finally we observe that for $p \geq 3, \frac{\partial^{p}}{\partial w^{p}} E_0(s,w)$ is a rational function of $f_0$ and its $w$-derivatives,
\begin{prop} \label{prop55}
\begin{eqnarray} \label{E0exp}
\frac{\partial^{p}}{\partial w^{p}} E_0(s,w) &=& {p \choose 2} \left( \frac{f_{0w}}{f_0}\right)_{w^{(p-3)}} + p w \left( \frac{f_{0w}}{f_0}\right)_{w^{(p-2)}} + \frac{w^2}{2} \left( \frac{f_{0w}}{f_0}\right)_{w^{(p-1)}}\\
\nonumber &+& \frac{(\nu-1)^2}{2\nu(\nu+1)} \left[\sum_{j=0}^{\lfloor \frac p2\rfloor} {p \choose j} f_{0 w^{(j)}} f_{0 w^{(p-j)}} - \frac{2\nu+1}{\nu-1}\left(wf_{0w^{(p)}} + p  f_{0 w^{(p-1)}}\right)\right]\\
\nonumber &+& (p-3)! \left(-\frac{1}{w}\right)^{p-2}
\end{eqnarray}
\end{prop}
Each line of the above proposition can be established directly by induction starting with the base case for $p=3$. 
It then follows from Proposition 3.1(iii) of \cite{Er09} that the minimal pole order contributed by (\ref{E0}) is $\geq 2g$. 

We are now in a position to outline the
\begin{proof} (of Theorem \ref{thm51})
In \cite{Er09} (Theorem 1.3) it was shown that 
\begin{eqnarray}\label{eg_ratl}
e_g(z_0) &=& \frac{(z_0-1) q_{d(g)}(z_0)}{(\nu - (\nu-1)z_0)^{o(g)}}
\end{eqnarray}
where $q_{d(g)}(z_0)$ denotes a polynomial of degree $d(g)$ in $z_0$. We first want to determine the relation between this degree and the pole order $o(g)$. To this end we observe from Propositions \ref{prop52}, \ref{prop53}, \ref{prop54}, and \ref{prop55} that the right hand side of 
(\ref{Hirota2}), evaluated at $w=1$,  is a rational function in $z_0$ which approaches a finite constant value as $z_0 \to \infty$.  From the form of the left hand side of (\ref{Hirota2}) evaluated at $w=1$ one also sees that its asymptotic order (as $z_0 \to \infty$) is the same as that of $e_g$. 
Hence, $d(g) = o(g) - 1$ and this shows that (\ref{note}) is valid up to the determination of the minimal and maximal pole orders at $z_0 = \nu/(\nu - 1) $. In the preceding lemmas and proposiitons we have seen that, for all terms on the right hand side of (\ref{Hirota2}), the minimal pole order is $\geq 2g$. Furthermore, from these same representations together with Proposition 3.1(iii) of \cite{Er09} one sees that, with the possible exception of the genus 0 terms in (\ref{E0exp}), the maximal pole order of the terms on the right hand side of (\ref{Hirota2}) is $5g-1$. The apparent maximal pole order in (\ref{E0exp}) is $4g+3$ which exceeds the stated bound when $g=2, 3$. This maximal order comes from terms containing the factor $f_{0 w^{(2g+2)}}$ which are, specifically,
\begin{eqnarray*}
&& \left[\frac12 \frac{f_{0 w^{(2g+2)}}}{f_0} - \frac{(2\nu+1)(\nu-1)}{2\nu(\nu+1)} f_{0 w^{(2g+2)}} + \frac{(\nu-1)^2}{2\nu (\nu+1)} \left(f_0 f_{0 w^{(2g+2)}}\right) \right]_{w=1}\\
&=& \frac{f_{0 w^{(2g+2)}}}{2f_0}|_{w=1} \left[1 - \frac{(2\nu+1)(\nu-1)}{2\nu(\nu+1)} z_0 +
 \frac{(\nu-1)^2}{2\nu (\nu+1)} z_0^2\right]\\
&=& \frac{f_{0 w^{(2g+2)}}}{2f_0}|_{w=1} \mathcal{O}(\nu - (\nu-1)z_0).
\end{eqnarray*} 
Hence the maximal pole order contributed by the genus 0 terms is, in fact, $4g+2$ which is 
$\leq 5g-1$ for $g>2$ and $< 5g-1$ for $g>3$. This establishes that, for $g > 2$, the pole orders on the right hand side of (\ref{Hirota2}) are bounded between $2g$ and $5g-1$. Moreover, for $g > 3$, the case by case checking of terms on the right hand side of (\ref{Hirota2}) that has been carried out in this subsection, shows that the maximal pole order is realized by the term in Proposition \ref{prop53} corresponding to the partition $\lambda$ of $2g$ having minimal length (=$1$); i.e., the partition whose Young diagram is a single row. This implies that the residue of the maximal order pole is $g!\, a_{3g-1}^{(g)}(\nu)$ which is non-vanishing by Theorem \ref{result}. Hence the maximal order pole is realized.  Now, given that $e_g$ has the form (\ref{eg_ratl}) with $d(g) = o(g) - 1$, it follows from direct calculation that $\frac{\partial^2}{\partial w^2}E_g(s,1)$ raises the minimum pole degree by $2$ and the maximum pole degree by $4$ with the coefficient at this order given by (\ref{leadcoeff}). This establishes (\ref{note}) for $g > 3$. The cases of $g=2,3$ may be established separately by direct calculation (see, for example, section 1.4.2 of \cite{Er09}). 

To establish (\ref{note3}) first note that by Euler's relation, $2-2g = m -\nu m +F$ for a $g$-map where $m$ is the number of ($2\nu$-valent) vertices and $F$ is the number of faces.  Since $F \geq 1$, one immediately sees that the number of vertices of such a map must satisfy the inequality
\begin{eqnarray*}
m \geq \frac{2g-1}{\nu-1}. 
\end{eqnarray*}
It follows that $e_g^{(j)}(s=0) = 0$ for $j \leq r = \max\left\{1, \frac{2g-1}{\nu-1}\right\}$. ($e_g$ must vanish at least simply at $s=0$ since since $\tau^2_n(s=0) \equiv 1$.) Via Cauchy's theorem these conditions may be re-expressed as
\begin{eqnarray*}
0 &=& \frac1{2\pi i} \oint_{s \sim 0} \frac{e_g(s)}{s^{j+1}} ds \\
&=& \frac{-1}{2\pi i} \oint_{z \sim 1} \frac{z^{\nu j - 1} q_{5g-6}(z)}{(\nu - (\nu-1)z)^{5g-6}(z-1)^j} dz
\end{eqnarray*}
for $j \leq r$ where in the second line we have rewritten $e_g$ as a rational function of $z$ (\ref{eg_ratl}) and employed the change of variables $\frac{ds}{dz} = - \frac{\nu - (\nu-1)z}{z^{\nu+1}}$ which may be deduced from the string equation (\ref{stringeqn}). This yields a contour integral in $z$ centerd at $1$. Now one can see that these vanishing conditions are satisfied if and only if
$q_{5g-6}^{(j)}(z=1) = 0$ for $j \leq r$ which in turn proves (\ref{note3}). 

Finally we turn to the determination of the constant $C^{(g)}$. By Proposition \ref{prop53}, contributions to the constant term of $e_g$ come only from the first sum on the right hand side of (\ref{Hirota2}). The parts of this coming from $g=0$ and $g=1$ are, by Propositions \ref{prop55} and \ref{prop54} respectively, $-2\frac{(2g-1)!}{(2g+2)!}$ and $2\frac{(2g-1)!}{(2g)!}\frac1{12}$. 
At higher genus, $k < g$, the contribution to the constant term is determined by Lemma \ref{lem51} to be $-2\frac{(2-2k)_{2g-2k+2}}{(2g-2k+2)!}C^{(k)}$. Hence, by (\ref{Hirota2}) we have
\begin{eqnarray*}
(2-2g)(1-2g) C^{(g)} &=& -2\frac{(2g-1)!}{(2g+2)!} + 2\frac{(2g-1)!}{(2g)!}\frac1{12} -2 \sum_{k=2}^{g-1} \frac{(2-2k)_{2g-2k+2}}{(2g-2k+2)!}C^{(k)},
\end{eqnarray*}
from which (\ref{note3}) immediately follows.
\end{proof}

\section{The Case of Odd Valence} \label{sec:6}
In the case when $j$ is odd in the weight (\ref{genpot}) for $V$, there is clearly a problem in applying the method of orthogonal polynomials as it was outlined in Section  \ref{sec:2}.
Very recently, however, a generalization of the {\it equilibrium measure} (which governs the leading order behavior of the free energy associated to (\ref{RMT}))  was developed and applied to this problem, \cite{BD10}. It is based on generalizing to a class of complex valued non-Hermitean orthogonal polynomials on a contour in the complex plane other than the real axis.  These extensions were motivated by new ideas in approximation theory related to complex Gaussian quadrature of integrals with high order stationary points \cite{DHK}. 

But even when the issue of existence of appropriate orthogonal polynomials has been resolved, there are still a number of significant obstacles to deriving results like Theorem \ref{thm51} that are not present when the valence $j$ is even.  For odd valence there is an additional string of recurrence coefficients, the diagonal coefficients $a_n$ of $\mathcal{L}$, whose asymptotics needs to be analyzed. This in turn requires that the lattice paths used to define and analyze the Toda and difference string equations must be generalized to the class of {\it Motzkin paths} which can have segments where the lattice site remains fixed rather than always taking a step (either up or down) as was the case for Dyck paths. 

Nevertheless, all these constructions have been carried out in \cite{EP11} to derive the hierarchies of continuum Toda and difference string equations when the valence $j$ is odd.  

The recurrence  coefficients again  have asymptotic expansions with continuum representations given by
\begin{align}
  a_{n+k, N} &= h(s_1, s_{2\nu+1},w) =  x^{1/2} \sum_{g \geq 0} h_g(s_1, s_{2\nu+1},w) n^{-g}\\
\nonumber  h_g(s_1, s_{2\nu+1}, w) &= -  w^{1-g} \times
\\  & \hspace{-0.5cm}
\sum_{\begin{matrix} 2 g_1 + j = g+1 \\  g_1 \geq 0\,, j>0\end{matrix}}  \frac1{j!} \frac{\partial^{j+1}}{\partial s_1 \partial \tilde{w}^j} \left[ \tilde{w}^{2-2g_1} e_{g_1}\left((w\tilde{w})^{-1/2} s_1, (w\tilde{w})^{\nu-1/2} s_{2\nu+1}\right)\right]_{\tilde{w} = 1}
\end{align}
The off-diagonal coefficients have corresponding representations which are much as they were in the even valence case,
\begin{align}
 b^2_{n+k,N} &= f(s_1, s_{2\nu+1},w) = x \sum_{g \geq 0} f_g(s_1, s_{2\nu+1},w)  n^{-2g}\\
 f_g(s_1, s_{2\nu+1}, w) &= w^{2-2g} \frac{\partial^2}{\partial s_1^2}  e_g(w^{-1/2} s_1, w^{\nu-1/2} s_{2\nu+1}). 
\end{align}
The coefficients in these expansions have a self-similar structure given by
\begin{align}
h_g(s_1, s_{2\nu+1},w) &= w^{\frac{1}{2} -g}  u_g(s_1 w^{-1/2}, s_{2\nu+1} w^{\nu - \frac{1}{2}})     \\
 f_g(s_1, s_{2\nu+1},w)  &= w^{1 -2g}  z_g(s_1 w^{-1/2}, s_{2\nu+1} w^{\nu - \frac{1}{2}}).
\end{align}

At leading order the continuum Toda equations are
\begin{eqnarray} \label{TODA}
\frac{\partial}{\partial s} 
\left(
\begin{array}{c}
h_0 \\ f_0
\end{array}
\right) &+& (2\nu + 1) 
\left(
\begin{array}{cc}
B_{11} & B_{12}\\ 
f_0 B_{12} & B_{11}
\end{array}
\right)  \frac{\partial}{\partial w} 
\left(
\begin{array}{c}
h_0 \\ f_0
\end{array}
\right) = 0
\end{eqnarray}
and the leading order continuum  difference string equations are
\begin{eqnarray}\label{STRING} 
s
\left(
\begin{array}{c}
0 \\ 1
\end{array}
\right) &=& 
\left(
\begin{array}{cc}
A_{11} & A_{12}\\ 
f_0 A_{12} & A_{11}
\end{array}
\right)  \frac{\partial}{\partial w} 
\left(
\begin{array}{c}
h_0 \\ f_0
\end{array}
\right)
\end{eqnarray}
where the coefficients of the matrix $\bf B$ in (\ref{TODA}) are specified by
\begin{eqnarray} \label{B11}
B_{11} &=& \sum_{\mu = 1}^\nu {2\nu \choose 2\mu - 1, \nu - \mu, \nu - \mu +1} h_0^{2\mu-1} f_0^{\nu-\mu+1}\\ \label{B12}
B_{12} &=& \sum_{\mu = 1}^\nu {2\nu \choose 2\mu, \nu - \mu, \nu - \mu} h_0^{2\mu} f_0^{\nu-\mu}\, ,
\end{eqnarray}
and those of the matrix $\bf A$ in (\ref{STRING}) by
\begin{eqnarray} \label{A11}
A_{11} &=& 1 + (2\nu+1)s \sum_{\mu = 0}^{\nu-1} {2\nu \choose 2\mu + 1, \nu - \mu - 1, \nu - \mu} (\nu - \mu) h_0^{2\mu+1} f_0^{\nu-\mu-1}\\ \label{A12}
A_{12} &=& (2\nu+1)s \sum_{\mu = 0}^{\nu-1} {2\nu \choose 2\mu, \nu - \mu -1, \nu - \mu +1} (\nu - \mu +1) h_0^{2\mu} f_0^{\nu-\mu -1}\, .
\end{eqnarray}
\begin{rem}
The index $\mu$ appearing in the trinomial coefficients corresponds to the number of flat steps in the Motzkin paths giving rise to that term.
\end{rem}

It is straightforward to see that (\ref{TODA}) may be rewritten in conservation law form as
\begin{eqnarray} \label{LAW}
\frac{\partial}{\partial s} 
\left(
\begin{array}{c}
h_0 \\ f_0 + \frac12 h_0^2
\end{array}
\right) &+& 
\frac{\partial}{\partial w} 
\left(
\begin{array}{c}
\Psi_1\\ 
\Psi_2 + h_0 \Psi_1
\end{array}
\right)  = 0
\end{eqnarray}
where the coefficients in the flux vector are given by
\begin{eqnarray} \label{flux1}
\Psi_{1} &=& \sum_{\mu = 0}^\nu {2\nu + 1\choose 2\mu, \nu - \mu, \nu - \mu +1} h_0^{2\mu} f_0^{\nu-\mu+1}\\
\label{flux2}\Psi_{2} &=&  \sum_{\mu = 0}^\nu {2\nu + 1 \choose 2\mu + 1, \nu - \mu -1, \nu - \mu +1} h_0^{2\mu + 1} f_0^{\nu-\mu + 1}\, .
\end{eqnarray}

Recently, \cite{EW}, we have determined that the equations (\ref{STRING}) are in fact a differentiated form of the generalized hodograph solution of the conservation law (\ref{LAW}). This hodograph solution is given by
\begin{eqnarray}\label{hodtostring1}
\Phi_1 &\doteq& h_0 + (2\nu+1) s B_{12} = 0\\
\label{hodtostring2}\Phi_2 &\doteq& f_0 + (2\nu+1) s B_{11} = w.
\end{eqnarray}

Analogous to what was done in Theorem \ref{thm51} we expect to determine closed form expressions for all the coefficients in the topological expansion with odd weights. The first few of these, for the trivalent case, are \cite{EP11} 
\begin{align} 
\nonumber e_0(t_3) &= \frac{1}{2} \log(z_0) + \frac{1}{12} \frac{ (z_0-1)(z_0^2-6 z_0 - 3)}{(z_0+1)}\,, \\
\label{egs} e_1(t_3) &= -\frac{1}{24} \log\left( \frac{3}{2} - \frac{z_0^2}{2} \right)\,, \\
\nonumber e_2(t_3) &= \frac{1}{960} \frac{ (z_0^2-1)^3 (4 z_0^4 - 93 z_0^2 -261)}{(z_0^2 - 3)^5}\,,
\end{align}
where $z_0$ is implicitly related to $t_3$ by the polynomial equation 
\begin{equation*}
1 = z_0^2 - 72 t_3^2 z_0^3 \,.
\end{equation*}
$z_0(t_3)$ is in fact the generating function for a {\it fractional} generalization of the Catalan numbers. Its $m^{th}$ coefficient counts the number of connected, non-crossing symmetric graphs on $2m+1$ equi-distributed vertices on the unit circle \cite{P}.
 
\section{Concluding Remarks}

\subsection{Spectrum} \label{71}
Nothing has been said, in this article, about the eigenvalues of the random matrix $M$  although this is at the heart of the Riemann Hilbert analysis underlying all of our results. The essential link comes through the {\it equilibrium measure} \cite{EM03, EMP08}, or density of states, for these eigenvalues. When $t_j = 0$ in (\ref{I.001b}), this equilibrium measure reduces to the well-known Wigner semi-circle law. As $t_j$ changes this measure deforms; but, for $t_j$ satisfying the bounds implicit in (\ref{I.002}) (i), its support remains a single interval, $[\alpha, \beta]$. The {\it edge of the spectrum} $\alpha(t_j), \beta(t_j)$ evolves dynamically with $t_j$. In fact we have determined \cite{EW} that $\alpha$ and $\beta$ are the Riemann invariants of the hyperbolic system (\ref{TODA}). In the case of even valence these invariants collapse to $\pm 2 \sqrt{z_0}$ where $z_0$ is the generating function upon which all the $e_g$ are built, as described in Theorem \ref{thm51}. So the edge of the spectrum is indeed directly related to the genus expansion (\ref{I.002}). 

\subsection{Random Surfaces} \label{72}
It follows from Corollary \ref{two-leg} and the map-theoretic interpretation of $e_g$ given by (\ref{gquotII}) that $z_g(t_j)$ is a generating function for enumerating $j$-regular $g$-maps with, in addition, two {\it legs}. A {\it leg} is a univalent vertex; i.e., a vertex with just one adjacent edge, that is  connected to some other vertex of the map.  In particular $z_0(t)$ enumerates such maps on the Riemann sphere or what are more commonly referred to as two-legged {\it planar maps}. In a remarkable paper, \cite{Sh},  building on prior work of Cori and Vauquelin \cite{CorVaq}, Schaeffer  found a constructive  correspondence between two-legged $2\nu$-regular planar maps and $2\nu$ valent {\it blossom trees}. A $2\nu$ valent blossom tree is a rooted $2\nu$ valent tree, with each external vertex taken to be a leaf that is colored either black or white and such that each internal (non-root) vertex is adjacent to exactly $\nu - 1$ black leaves. This gives another interpretation of $z_0(t_{2\nu})$ as the generating function for the enumeration of blossom trees. It seems reasonable to hope that the arithmetic data implicit in the coefficients $a^{(g)}_m(\nu)$ in Theorem \ref{result} (resp. $c^{(g)}_m(\nu)$ in Theorem \ref{thm51}) might provide a means, such as {\it sewing rules}, for constructing two-legged $2\nu$-regular $g$-maps (resp. pure $2\nu$-regular $g$-maps) from blossom trees. 

 In another direction Bouttier, Di Francesco and Guitter \cite{dF} have studied the combinatorics of geodesic distance for planar maps. They define the geodesic distance of a two-legged graph to be the minimum number of edges crossed by a continuous path between the two legs and study $r_d(t_{2\nu})$, the generating function for enumerating all two-legged $2\nu$ valent planar maps whose geodesic distance is $\leq d$. They find surprising and elegant closed form expressions for the $r_d(t_{2\nu})$. The statisitcs of planar maps is a natural stepping off point for the study of random surfaces. There has been a lot of recent activity in this direction by Le Gall and his collaborators related to the work of Schaeffer and Bouttier et al. See for example \cite{LG}.

\subsection{Enumerative Geometry of Moduli Spaces} \label{73}
A different (from matrix models) representation of 2D Quantum Gravity may be given in terms of intersection theory on the moduli space of stable curves (Riemann surfaces),  $\overline{\mathcal{M}}_{g,n}$ and from this alternate perspective Witten conjectured that a generating function for the intersection numbers of tautological bundles on $\overline{\mathcal{M}}_{g,n}$ should be "given by" a double-scaling limit of the differentiated free energy (\ref{b-asymp}) for matrix models. (A precise description of this double-scaling limit may be found in \cite{Er09}.) He further conjectured that this intersection theoretic generating function should, with repect to appropriate choices of parameter vatriables, satisfy the Korteweg-deVries (KdV) equation. 
Subsequently, Kontsevich \cite{K} was able to outline a proof of Witten's KdV conjecture based on a combinatorial model of intersection theory on $\overline{\mathcal{M}}_{g,n}$. This model expresses tautological intersections in terms of sums over tri-valent graphs on a genus $g$ Riemann surface. He was then able to recast this sum in terms of a special matrix integral involving cubic weights on which the proof of Witten's KdV conjecture is based. For a readable overview on the above circle of ideas we refer to \cite{OP}.

However, the first Witten conjecture, on relating the intersection-theoretic free energy to the matrix model free energy, (\ref{I.002}), remains open. With the results described in this paper it may now be possible to determine if, and in precisely what sense, this conjecture might be true and to see if this leads to connections between the KdV equation in the Witten-Kontsevich model and the conservation laws given by (\ref{TODA}).  In addition, given the recent results on matrix models with odd dominant weights, \cite{BD10, EP11}, it may now be possible to give a rigorous treatment of Kontsevich's matrix integral which, up to now, has been formal. 

More recently there have been other, perhaps more natural, approaches to the proof of Witten's KdV conjecture, \cite{KL, Mir, GJV}, in terms of coverings of the Riemann sphere and {\it Hurwitz numbers} for which the generating functions specified in Theorem {\ref{thm51}} should also have a natural interpetation.

\subsection{Analytical Deformations and Critical Parameters} \label{74}
In \cite{Er09} it was observed that the equilibrium measure (see \ref{71}) for the weight $V$, with $j=2\nu$ in (\ref{I.001b}), may be re-expressed as 
\begin{eqnarray} \label{eqmeas}
\mu_{V_{t}}(\lambda) &=& z_0 \mu_0(\lambda) + (1 - z_0) \mu_\infty(\lambda)
\end{eqnarray}
where $\mu_0$ is the equilibrium measure for $V = 1/2 \lambda^2$ (the {\it semicircle} law) and $\mu_\infty$ is the equilibrium measure for $V = \lambda^{2\nu}$; i.e., the general measure is a linear combination, over $z_0$, of two extremal monomial equilibrium measures. For $z_0 \in [0, 1]$ (which corresponds to $t_{2\nu} \in [0, \infty]$), this combination is convex and (\ref{eqmeas}) is indeed a measure with a single interval of support in $\lambda \in \mathbb{R}$. This may be analytically continued to a complex $z_0$ neighborhood of $[0,1]$ so that (\ref{eqmeas}) remains a positive measure along an appropriate connected contour ("single interval") in the complex $\lambda$-plane. For $\nu = 2$ this continuation may be made up to a boundary curve in the complex $z_0$-plane passing through $z_0 = 2$ (with a corresponding image in the complex $t_4$ plane). Extension of this to more general values of $\nu$ is in progress \cite{EM12}. The mechanism for carrying out this continuation is to regard (\ref{eqmeas}) as a $z_0$-parametrized family of holomorphic quadratic differentials. The candidate for the measure's support is then an appropriate bounded real trajectory of the quadratic differential.  Outside the boundary curve, the Riemann-Hilbert analysis used in this paper may be analytically deformed and our results extended. The boundary may be regarded as a curve of critical parameters for this deformation. This curve is precisely the  locus where the Riemann invariants, that determine the edge of the spectrum (as described in \ref{71}) exhibit a shock.   

This scenario is reminiscent of that for the small $\hbar$-limit of the nonlinear Schr\"odinger equation \cite{JLM, KMM} in which the analogue of our boundary curve is the envelope of {\it dispersive shocks}.  In that setting it is the Zakharov-Shabat inverse scattering problem that shows one how to pass through the dispersive shocks and describe a continuation of measure-valued solutions with so-called {\it multi-gap} support. It is our expectation that coupling gravity to an appropriate conformal field theory (to thus arrive at a bona fide string theory) \cite{Mar} will play a similar role in our setting to determine a unique continuation through the boundary curve of critical parameters to a unique equilibrium measure with multi-cut support. We also hope that this will help bring powerful methods from the study of dispersive limits of nonlinear PDE into the realm of random matrix theory. 
\medskip

\noindent {\bf Acknowledgement.}  The author wishes to thank MSRI for its hospitality and the organizers for the excellent Fall 2010 program on Random Matrix Theory. Most of the new results described here had their inception during that happy period.

\end{document}